\documentclass[12pt, draftclsnofoot, onecolumn]{IEEEtran}

\usepackage{cite}
\usepackage{amsfonts}
\usepackage{amssymb}
\usepackage{stfloats}
\usepackage{graphicx}
\usepackage{subfigure}
\usepackage{amsmath}
\usepackage{array}
\usepackage{epstopdf}
\usepackage{graphicx} 
\usepackage{amsthm} 
\usepackage{algorithm}
\usepackage{setspace}
\usepackage{algpseudocode}
\usepackage{amsmath}
\usepackage{graphics}
\usepackage{epsfig}
\usepackage{setspace}
\usepackage{extarrows}
\usepackage[usenames]{color}

\IEEEoverridecommandlockouts

\newtheorem{Definition}{Definition}
\newtheorem{Problem}{Problem}
\newtheorem{Lemma}{Lemma}

\newtheorem{Algorithm}{Algorithm}
\newtheorem{Policy}{Policy}

\newtheorem{Remark}{Remark}

\newtheorem{Baseline}{Baseline}

\title{Joint Optimization of File Placement and Delivery in Cache-Assisted Wireless Networks with Limited Lifetime and Cache Space}
\author{

	\IEEEauthorblockN{Bojie Lv, Rui Wang, Ying Cui, Yi Gong and Haisheng Tan}

\thanks{
Bojie Lv, Rui Wang and Yi Gong are with the Department of Electrical and Electronic Engineering, Southern University of Science and Technology, China. Ying Cui is with the Department of Electronic Engineering, Shanghai Jiao Tong University, China. Haisheng Tan is with the School of Computer Science and Technology, University of Science and Technology of China.
	
Part of this work has been accepted in IEEE GLOBECOM 2018 \cite{Ruiwang2018_GLOBECOM}. We have extended the conference paper by revising the low-complexity scheduling policy design in Section \ref{sec:approximation}, developing a novel reinforcement learning algorithm in Section \ref{sec:learning}, improving the bounds on the approximate value functions in Section \ref{Sub:Sec:bounds}, and generating more illustrative simulation results.} 
} 
 
\begin{document}
\columnsep 0.22in

\maketitle
\begin{abstract}
In this paper, the scheduling of downlink file transmission in one cell with the assistance of cache nodes with finite cache space is studied. Specifically, requesting users arrive randomly and the base station (BS) reactively multicasts files to the requesting users and selected cache nodes. The latter can offload the traffic in their coverage areas from the BS. We consider the joint optimization of the abovementioned file placement and delivery within a finite lifetime subject to the cache space constraint. Within the lifetime, the allocation of multicast power and symbol number for each file transmission at the BS is formulated as a dynamic programming problem with a random stage number. Note that there are no existing solutions to this problem. We develop an asymptotically optimal solution framework by transforming the original problem to an equivalent finite-horizon Markov decision process (MDP) with a fixed stage number. A novel approximation approach is then proposed to address the curse of dimensionality, where the analytical expressions of approximate value functions are provided. We also derive analytical bounds on the exact value function and approximation error. The approximate value functions depend on some system statistics, e.g., requesting users' distribution. One reinforcement learning algorithm is proposed for the scenario where these statistics are unknown.
\end{abstract}

\section{introduction}\label{sec:intro}

Caching is a promising technology to save the transmission resource and improve the spectrum efficiency for cellular networks. In this paper, we consider a flexible deployment scenario of a cache-enabled cell where there is no wired connection or dedicated spectrum between the base station (BS) and cache nodes. Thus the cache nodes and requesting users receive popular files simultaneously via downlink multicast. Moreover, the timeliness of popular files is considered as in \cite{Leconte2016}, and the transmission within a finite file lifetime is optimized via the approach of finite-horizon {\em Markov decision process (MDP)} and {\em reinforcement learning}.

\subsection{Related Works}

There have been a number of works on the optimization of file placement with the cache space constraint. It's obvious that cache nodes should store the most popular files if each user can get access to only one cache node. However, the papers \cite{W.Choi2016,K.B.Huang2017} showed that caching files randomly with optimized caching probabilities is better than storing the most popular files at each cache node when one user can be served by multiple cache nodes. In \cite{JunZhang2016}, the authors took user mobility into consideration, where each user can move among the service regions of different cache nodes. They proposed a file placement policy to improve data offloading rate. 
With the help of historical file request information, the authors in \cite{Shin2019} proposed a file placement and update method at the cache node via predicting the arrival distribution of future file requests. The optimal file placement strategies were designed in \cite{cui2019-TCOM} in the cases of imperfect and unknown file popularity distributions.  In \cite{tao2019}, the authors considered unmanned aerial vehicles as the users and designed a probabilistic file placement method to maximize the average successful file download rate.
Moreover,  there are also some works on the design of coded caching schemes \cite{coded_cache_1,M.Tao2017}. With cached files, the authors in \cite{Tao2016} designed a multicast beamforming policy to minimize the weight sum of the backhaul cost and transmit power at the BS, and the paper \cite{cui2016} formulated the joint minimization of the average delay and power consumption at the BS as a stochastic optimization problem. In all the above works, the cost of file placement at cache nodes is not taken into consideration, as it is assumed to be completed before the phase of file delivery to the requesting users. In practice, however, there may not be sufficient time for file placement before users' requests, e.g., real-time news. 

When the phases of file placement (at cache nodes) and delivery (to requesting users) occur simultaneously, joint scheduling of both phases becomes necessary. For example, a file placement and delivery framework for heterogeneous networks was investigated in \cite{Ansari2016}, where cache node association of requesting users and coded file placement are jointly optimized to maximize the overall throughput in each frame. In \cite{cui2016gc}, an optimal caching and user association policy was proposed to minimize the latency in a cached-enabled heterogeneous network with wireless backhaul. In the above works, the files are delivered to small BSs via dedicated backhaul links, i.e., there is no resource sharing between file placement and delivery. When there is no dedicated link or period for file placement at cache nodes, file placement and delivery can be simultaneously conducted using multicast \cite{Cui2017}. This yields a coupling relationship between the transmission resource consumption and file placement. For example, if more resource is spent on downlink multicast, files will be cached in more cache nodes, which may save the downlink resource in future transmissions. As a result, a joint optimization of file placement and delivery with the consideration of the total transmission resource consumption at the BS becomes inevitable. Moreover, it is of practical value to model the file requests as a temporal and spatial random process. Hence, dynamic programming can be utilized to address the joint optimization of file placement and delivery. This issue was initially studied in our previous work \cite{Lv2019}. Specifically, we considered a random number of  requests on multiple popular files without cache space limitation in \cite{Lv2019}, where the scheduling design for multiple files can be equivalently decoupled as single-file scheduling problem. The multi-file case with limited cache space at the cache nodes has not addressed.

Dynamic programming via MDP has been considered in resource allocation of wireless systems \cite{Moghadari2013,cui2010,Dechene2010,Wang2013,RuiWang2011,cui2012-TIT,Han2018,huang2019mdpbased} or information systems \cite{Han2016,Han2019,lv2019cooperative}. For example, infinite-horizon MDP was used to optimize the cellular uplink transmissions \cite{Moghadari2013,cui2010}, downlink transmissions \cite{Dechene2010}, and relay networks \cite{Wang2013,RuiWang2011}, where the average transmission delay is either minimized or constrained. Moreover, low-complexity solutions were considered in the abovementioned works to avoid the curse of dimensionality \cite{Shewhart2011Approximate}. Note that popular files to be stored at cache nodes usually have a finite lifetime, and hence infinite-horizon MDP adopted in the aforementioned works may not be suitable for joint optimization of file placement and delivery anymore. Nevertheless, the finite-horizon MDP is usually more complicated \cite{FHMDP}, and designing low-complexity algorithrms for finite-horizon MDP is still an open issue.

\subsection{Our Contributions}

In this paper, we consider the downlink transmission of popular files in a cache-enabled cell within a finite file lifetime. The popular files may not be stored in the cache nodes at the very beginning of the lifetime. The arrival of requesting users is random in both temporal and spatial dimensions. When one file is requested, the BS reactively multicasts it to the requesting user (a.k.a. file delivery) as well as some chosen cache nodes (a.k.a. file placement) according to the channel and cache status. With the decoded files, the cache nodes can serve the following requesting users in its coverage region via different spectrum from the downlink (e.g., Wi-Fi) as in \cite{May2016,Molisch2016}. Therefore, the current file transmission may lead to the update of the cache status, which affects the future file transmissions in the remaining lifetime. In this paper, the main contributions on the optimization and analysis of the transmission scheduling are summarized below.
\begin{itemize}
\item We consider the joint optimization of the file placement and delivery within a finite lifetime subject to the cache space constraint, and propose a novel optimization framework. In particular, we formulate the scheduling of multicast power and symbol number at the BS as a dynamic programming problem, where the goal is to minimize the average transmission cost (weighted sum of the transmission energy and symbol number) at the BS by offloading the traffic to the cache nodes. Due to the cache space limitation, less popular files stored at one cache node in the early stage of the lifetime and may be replaced with more popular files. This complicated replacement has not been address in our previous work \cite{Lv2019}, which focused on the transmission scheduling of one file only. 

\item Since the number of file requests in the lifetime is random, the dynamic programming problem formulation has a random stage number, and there are no existing solutions to this problem to our best knowledge. The main difficulty is that the remaining number of stages in the Bellman's equations is unknown. We address this issue by proposing a novel framework to  equivalently transform the original problem to a finite-horizon MDP with a fixed number of stages. In \cite{Lv2019}, however, the dynamic programming problem with a random stage number is transformed to a finite-horizon MDP in a heuristic way, losing the optimality.

\item In order to address the curse of dimensionality, a novel approach of value function approximation is proposed for the abovementioned finite-horizon MDP, where the approximate value functions can be calculated using analytical expressions efficiently and effectively. Instead of numerical algorithms that are computationally expensive. We also provide tight analytical upper and lower bounds on the exact value function (which represent the minimum average transmission cost). The approach of value function approximation in \cite{Lv2019} cannot be applied in this paper, due to the different definitions of value functions. Moreover, it is difficult to obtain an analytical upper bound on the minimum average transmission cost via the approach in \cite{Lv2019}.

\item The expressions of approximate value functions rely on some system statistics, e.g., the distribution of requesting users and the popularities of files. In the case where the priori knowledge on these system statistics is not available, a novel reinforcement learning algorithm is proposed to evaluate the approximate value functions in an online manner. The issue of unknown file popularity is not addressed in \cite{Lv2019}.

\end{itemize}

It is shown by simulations that, compared with some baseline schemes, the proposed low-complexity algorithm based on value function approximation can significantly reduce the average transmission cost at the BS.

\begin{figure}[tb]
	\centering
	\includegraphics[height=200pt,width=400pt]{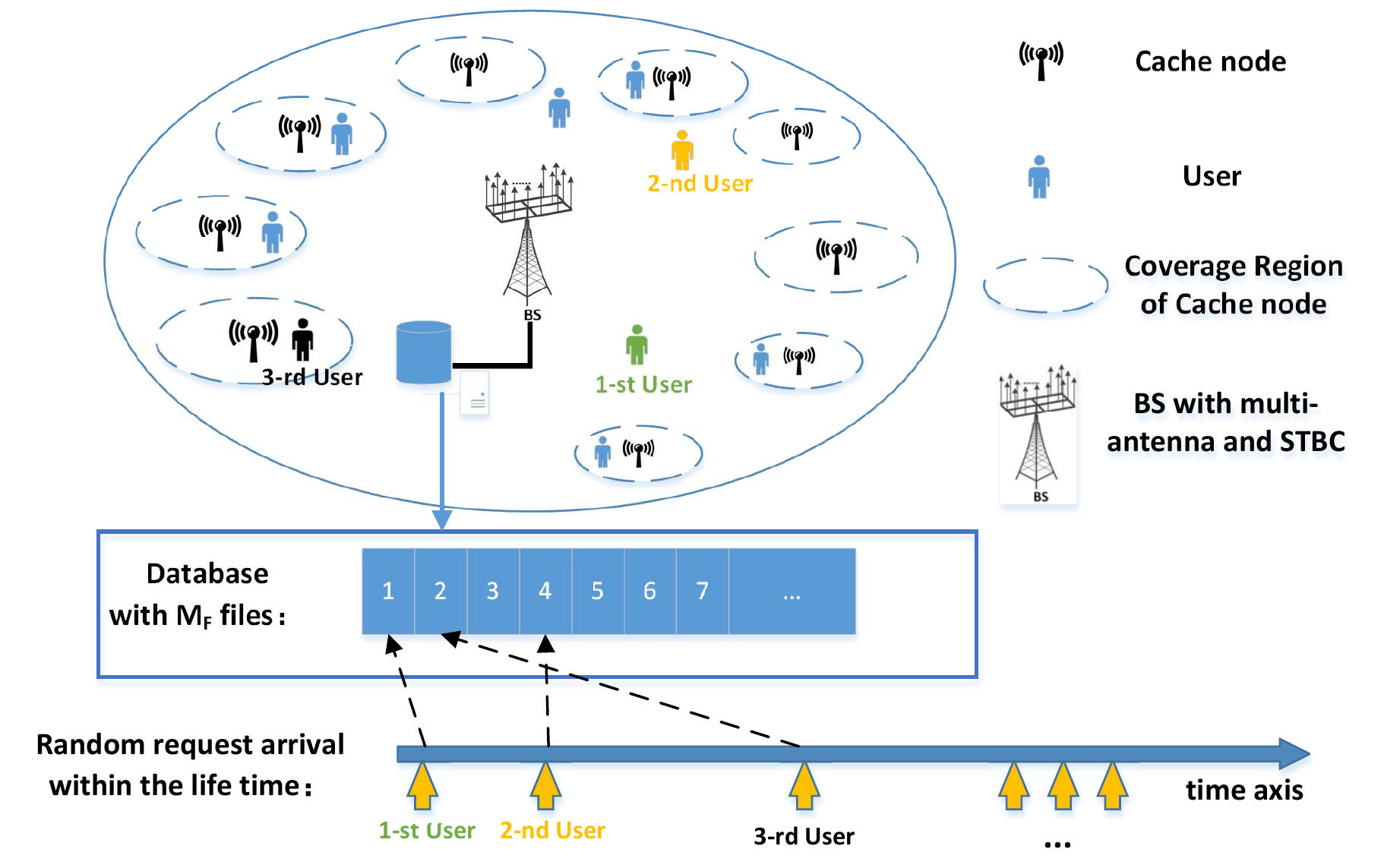}
	\caption{Illustration of network model with one BS, multiple wireless cache nodes and a file database, where one example of random spatial and temporal arrivals of requesting users is provided.}
	\label{fig:scheme}
\end{figure}

\section{System Model}\label{sec:model}

\subsection{Cache-Enabled Network Model}
 
As illustrated in Fig. \ref{fig:scheme}, we consider the downlink file transmission in a cell with one BS and $N_C $ cache nodes. The BS has $ N_T $ transmission antennas. Each cache node has a receiving antenna. Let $\mathcal{C}\subset \mathbb{R}^2$ be the service region of the cell, $ \mathcal{C}_c $ be the service region of the $ c $-th cache node  and $ \mathcal{C}_0 \triangleq \mathcal{C} -\mathcal{C}^*  $ be the region not served by any cache node, where $c=1,2,...,N_C$ and $ \mathcal{C}^* \triangleq \bigcup \limits_{c=1}^{N_C}\mathcal{C}_c$. It is assumed that the service regions of any two cache nodes are not overlapped, i.e., $ \mathcal{C}_i \cap  \mathcal{C}_j = \emptyset$ for all $ i \neq j$, $i,j=1,...,N_C $. A library with $ M_F $ files, denoted as $ \mathcal{F} \triangleq \{1,2,...,M_F\}$,  is accessible to the BS. For convenience of illustration, it is assumed that each file consists of $ R_F $ information bits, and the $c$-th cache node can store at most $ M_c $ files ($ M_c \leq M_F$). Our proposed algorithms can be easily extended to the scenario with different file sizes. In addition, it is assumed that there is at most one user requesting a file in each frame. The locations of the requesting users are independent and identically distributed (i.i.d.) in the cell according to a certain spatial distribution $ \mathcal{D} $. The distribution intensity at location $ \mathbf{l} \in \mathcal{C}$ is denoted as $ \rho_{\mathcal{D}}(\mathbf{l}) $, and the probability that a requesting user falls in an area $ \mathcal{S} \subseteq \mathcal{C} $ is $\int_{\mathcal{S}} \rho_{\mathcal{D}}(\mathbf{l}) d s(\mathbf{l}). $

The probability of the $ f $-th file being requested by one user in each frame is $ \beta p_f$, where $ \beta \in [0,1]$ is the probability that there is one user request in one frame,\footnote{in practice, the transmission of one file usually lasts over a large number of frames. The probability of more than one new request per frame is assumed to be negligible. Otherwise, the traffic of the network cannot be stabilized.} and $ p_f \in (0,1)$ is the probability that the requested file is the $f$-th one (the popularity of the $ f $-th file). Note that  $ \sum_{f=1}^{M_F} p_f = 1 $. Without loss of generality, we assume $ 1 \geq p_1 \geq p_2 \geq ... \geq p_{M_F} \geq 0$. In this paper, we do not have any restriction on the popularity distribution $ \{p_f|\forall f\} $. For example, it can be a Zipf distribution as in \cite{Web_cache}.

We consider a finite common lifetime $\mathfrak{L}$ for the file library, which usually lasts over several hours. The concept of lifetime captures the practical scenario where the popularity of a file (e.g. video news) may drop down quickly after a certain period. Suppose that there are $ L $ frames in the common file lifetime $ \mathfrak{L} $, i.e. $ \mathfrak{L} \triangleq \{1,2,...,L\} $. $ L $ is usually large. Let ${N}_R$ be the total number of file requests during the lifetime $ \mathfrak{L} $, which is a random variable with probability mass function (PMF) $
\Pr({N}_R=n)=\binom{L}{n}\beta^n (1-\beta)^{L-n}$, $ \forall n=0, 1,..., L. $

 At the beginning of the lifetime $ \mathfrak{L} $, the cache nodes may be empty or have stored some of the files. We have no requirement on the cache status at the beginning of the lifetime. A requesting user will download the requested file from one cache node (corresponding to traffic offloading from the BS) if it is in the coverage region of that cache node and the requested file has already been stored there; otherwise, it will download the file from the BS. In order to save downlink transmission resource, the BS can simultaneously transmit (i.e. multicast) the requested file to the user and some cache nodes. Hence, there are two types of file transmissions in the network, namely {\em BS multicast} and {\em cache node unicast}. The former refers to the downlink file delivery and placement from the BS to both a requesting user and possibly some selected cache nodes (which will be optimized later), when the requested file cannot be obtained from a cache node. The later is for the file delivery from any cache node to a requesting user using Wi-Fi, bluetooth, or other air interfaces, which is in different spectrum from the downlink transmission \cite{May2016,Molisch2016}. Compared with many existing literature on static caching, the receiving cache nodes optimization of BS multicast in the whole lifetime $ \mathfrak{L} $ is a dynamic caching problem with a random number of stages (requests): the cached files will be updated after each multicast, and the status of cached files affects the future transmission cost. 
 
\begin{figure}[tb]
	\centering
	\includegraphics[height=160pt,width=330pt]{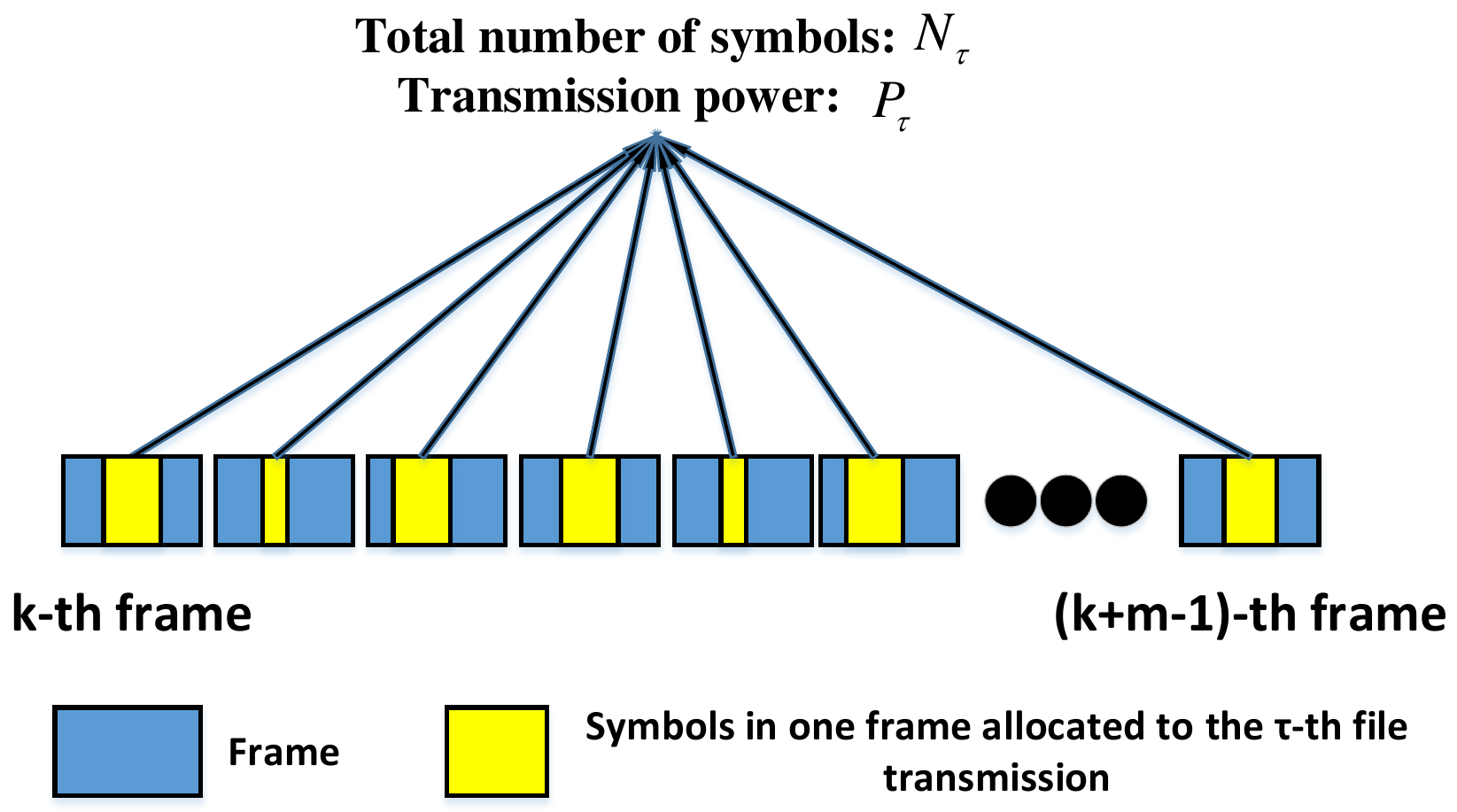}
	\caption{Illustration of multicast frame allocation for one file starting at the k-th frame.}
	\label{fig:frame}
\end{figure}

\subsection{Physical Layer Model of BS Multicast}

In this paper, it is assumed that the file size $ R_F $ is large, and the downlink transmission of each file is over a large number of physical-layer frames (we shall refer to the physical-layer frame as frame in the remaining of this paper). In practical systems, the downlink frame is not dedicated for one popular file's multicast. It may also carry transmission symbols of other traffics as illustrated in Fig. \ref{fig:frame}. In this paper, we do not specifiy the frame allocation of each file multicast. Instead, we focus on the scheduling of total number of symbols for each file multicast, which is in a larger time granularity than frame allocation. Because the file multicast consumes a large number of frames, the ergodic capacity averaged over small-scale channel fading can be achieved at each receiver. At the BS, each requested file is encoded in a rateless manner to an arbitrary number of modulation symbols. The BS determines the power and the number of modulation symbols to be sent for the multicast of each requested file. Let $ N_{\tau} $ be the number of downlink multicast symbols allocated for the $ \tau $-th file request, and $ P_{\tau} $ be the multicast power of these symbols.  
The following peak power constraint shall be satisfied.
\begin{align}\label{eqn:peak_power}
\textbf{Peak power constraint: }P_\tau \leq P_B, \forall \tau,
\end{align}
where $P_B$ is the maximum transmission power for file multicast at the BS.

The requesting user or a cache node is able to decode the file as long as its corresponding ergodic capacity of $N_{\tau}$ modulated symbols is greater than the file size $ R_F $. In this paper, we do not have any restriction on the file transmission  delay, i.e, there is no upper-bound constraint on $ N_{\tau} $ ($ \forall \tau $).

The space-time block code (STBC) with full diversity (e.g., Alamouti code) is used in the physical layer for BS multicast for the following two reasons. Firstly, it does not rely on the channel state information (CSI) at the transmitter (CSIT). Secondly, the diversity gain can be achieved at all the receivers. We refer to the user sending the $\tau$-th request as the $\tau$-th user. Let $ \rho_{\tau} $ and $ \rho^c $ be the pathlosses from the BS to the $ \tau $-th user and the $ c $-th cache node respectively, and let $ \eta_{\tau} $ and $ \eta_{\tau}^c $ be the corresponding shadowing coefficients during the $\tau$-th file transmission. Assume that $ \rho_{\tau} $, $ \eta_{\tau} $ and $ \eta_{\tau}^c $ are quasi-static within one file transmission, and change independently and identically over different transmissions. Following the capacity of full diversity
	STBC in \cite{paulraj2003introduction}, the maximum number of information bits that can be delivered from the BS to the $ \tau $-th user given transmission parameters ($P_{\tau},N_{\tau}$) is written as
\begin{equation}\label{eqn:dl-rate}
	R_{\tau} =  N_{\tau} \mathbb{E}_{\mathbf{h}_{\tau} }\left[ \alpha\log_2 \left( 1 + \frac{||\mathbf{h}_{\tau}||^2 P_{\tau}}{N_T \sigma^2_z} \right) \right],
\end{equation}
where $ \alpha $ is the code rate of the adopted full-diversity STBC. For example, when $ N_T = 2 $ and the Alamouti code is used, $ \alpha=1 $; $ \alpha $ is usually less than $ 1 $ for $ N_T>2 $. $\sigma^2_z$ is the average power of noise and inter-cell interference, and $ \mathbf{h}_{\tau}\in \mathbb{C}^{N_T} $ is the i.i.d. channel vector from the BS to the $\tau$-th user. Each element of $ \mathbf{h}_{\tau}\in \mathbb{C}^{N_T} $ follows a complex Gaussian distribution with zero mean and variance $ \rho_{\tau}\eta_{\tau} $. Equation \eqref{eqn:dl-rate} is the ergodic channel capacity over a large number of frames, the randomness in small-scale fading is then averaged. $ R_{\tau} $ depends only on the transmission parameters ($ P_{\tau}, N_{\tau} $) and the large-scale fading coefficient $ \rho_{\tau} \eta_{\tau} $. Hence, the following decoding constraint should be satisfied at the
$ \tau $-th user if it cannot receive the requested file from any cache node.
\begin{eqnarray} \label{constrain:user}
\mbox{\bf Downlink decoding constraint: }\ R_{\tau} \geq R_F, \ \forall \tau.
\end{eqnarray}

Similarly, the maximum number of information bits that can be delivered from the BS to the $ c $-th cache node given transmission parameters ($P_{\tau},N_{\tau}$) is written as $$R_{\tau}^c \!=\! N_{\tau} \mathbb{E}_{\mathbf{h}^c_{\tau}}\! \left[ \alpha\log_2 \left( \!1 \!+\! \frac{||\mathbf{h}^c_{\tau}||^2 P_{\tau}}{N_T \sigma^2_z} \!\right) \!\right],$$
where $ \mathbf{h}_{\tau}^c \in \mathbb{C}^{N_T}$ is the i.i.d. channel vector from the BS to the $ c $-th cache node with each element following the complex Gaussian distribution with zero mean and variance $ \rho^{c}\eta_{\tau}^c $. $ R_{\tau}^c $ depends only on the transmission parameters ($ P_{\tau}, N_{\tau} $) and the large-scale fading coefficient $ \rho^c \eta_{\tau}^c $. The $ c $-th cache node can decode the file only when  
$$
		R_{\tau}^c \geq R_F.
$$ 
		
The expressions of $R_{\tau}$ and $R_{\tau}^c$ depend on the pathloss and shadowing of the corresponding link. Hence, to decode one file in a BS multicast, the requesting user and cache nodes may need to accumulate different numbers of multicast symbols. By adjusting $ P_{\tau} $ and $ N_{\tau} $, the BS can control the set of receiving cache nodes. Moreover, different frame allocation schemes for these $ N_{\tau} $ symbols lead to different multicast transmission delay. For example, the transmission delay is large if each frame carries a small number of multicast symbols. We do not directly optimize the delay performance in this work since we focus on a scheduling granularity larger than frame.

\subsection{Cache Dynamics}

Let $ \{\mathcal{B}_{f,\tau}^{c}| \forall f \} $  be the cache state information (CaSI) of the $ c $-th cache node at the beginning of the $ \tau $-th file request, where $  \mathcal{B}_{f,\tau}^c \in\{0,1\}$. $ \mathcal{B}_{f,\tau}^c  =1 $ means that the $ f $-th file has been stored at the $ c $-th cache node before the $ \tau $-th request, and $ \mathcal{B}_{f,\tau}^c =0 $ otherwise.
 Let $\Delta\mathcal{B}_{f,\tau}^c = \mathcal{B}_{f,\tau+1}^c  - \mathcal{B}_{f,\tau}^c$, be the corresponding update of the CaSI for the $f$-th file at the $ c $-th cache node after the transmission of the $ \tau $-th requested file. $ \Delta\mathcal{B}_{f,\tau}^c = -1 $ means that the $ f $-th file cached at the $ c $-th cache node is replaced after the transmission of the $ \tau $-th requested file, and $ \Delta\mathcal{B}_{f,\tau}^c = 1 $ means that the $ f $-th file is cached at the $ c $-th cache node after the transmission of the $ \tau $-th requested file. Hence, we have the following constraint due to limited cache size.
\begin{eqnarray}\label{constrain:cache}
\mbox{\bf Cache size constraint: }\ \sum \limits_{f=1}^{M_F}\bigg( \mathcal{B}_{f,\tau}^c+\Delta\mathcal{B}_{f,\tau}^c \bigg) \leq M_c, \ \forall c,\tau.
\end{eqnarray}
Moreover, letting $ \mathcal{A}_{\tau} $ be the index of the $ \tau $-th requested file, we have the following constraints on the update of CaSI.
\begin{align}
\mbox{\bf CaSI update on the decoded file: }\  &\Delta \mathcal{B}_{\mathcal{A}_{\tau},\tau}^{c} \in \{0, \mathbf{I}[R_{\tau}^c\geq R_F] \}, \ \forall c,\tau ,\label{constrain:cache_action1} \\
\mbox{\bf CaSI update on the cached files: }\ 
 &\Delta \mathcal{B}_{f,\tau}^c \in \{0,- \mathcal{B}_{f,\tau}^c\}, \forall c,\tau, f \neq \mathcal{A}_{\tau}, \label{constrain:cache_action2} 
\end{align}
where $\mathbf{I}(.)$ denotes the indicator function. Note that (\ref{constrain:cache_action1}) is about the decision on whether to cache the decoded file. For instance, if the $ c $-th cache node is able to decode the multicasted file ($ \mathbf{I}[R_{\tau}^c\geq R_F] =1 $), it may store the file ($ \Delta \mathcal{B}_{\mathcal{A}_{\tau},\tau}^c=1 $) for possible later transmissions or discard the file ($ \Delta \mathcal{B}_{\mathcal{A}_{\tau},\tau}^c=0 $) due to the cache size constraint in (\ref{constrain:cache}). Moreover, (\ref{constrain:cache_action2}) is about the decision on whether to remove one file from a cache node due to limited cache space.
	
\subsection{System State and Scheduling Policy}

When a requesting user (say the $ \tau  $-th user) cannot be served by its nearby cache node, the BS should determine the downlink multicast power $ P_{\tau} $ and the number of transmission symbols $ N_{\tau} $ for BS multicast. In order to formulate the downlink scheduling problem, we first define the system state $ S $ and scheduling policy $ \Omega $ as follows.

\begin{Definition}[System State]
At the $ \tau $-th request arrival, the system state is represented by $ S_{\tau}\triangleq (\mathcal{A}_{\tau}, B_{\tau}, \zeta_{\tau})$, consisting of
\begin{itemize}
	\item Index of the requested file: $\mathcal{A}_{\tau} \in \mathcal{F} $.
	
	\item CaSI: $B_{\tau} \triangleq \{\mathcal{B}_{f,\tau}^c|\forall f\in \mathcal{F}, c=1,...,N_C \}$.
	\item Statistical channel state information (SCSI): the pathloss and shadowing coefficients of the channel from the BS to the $ \tau $-th requesting user and the shadowing coefficients of the channels from the BS to all the cache nodes, denoted as $\zeta_{\tau}  \triangleq \{(\eta_{\tau},\eta_{\tau}^c,\rho_{\tau})|\forall c=1,...,N_C\}$. 
\end{itemize}
\end{Definition}

\begin{Definition}[Scheduling Policy] \label{def:policy}
	At the $ \tau $-th request arrival, the scheduling policy $ \Omega_{\tau} $ for the $\tau$-th requesting user is a mapping from the system state $S_{\tau}$ to the following scheduling actions: BS multicast power $ P_{\tau} $ and symbol number $ N_{\tau} $; cache update $  \{\Delta\mathcal{B}_{f,\tau}^c|\forall f \in \mathcal{F}, c=1,2,...,N_C  \} $. 
	Meanwhile, the constraint (\ref{eqn:peak_power}) on peak power, the constraint (\ref{constrain:user}) on successful decoding at the $\tau$-th requesting user, cache size constraint in (\ref{constrain:cache}), and cache update constraints in (\ref{constrain:cache_action1})-(\ref{constrain:cache_action2}) should be satisfied.
\end{Definition}
Given a scheduling policy, the system state evolves as a Markov chain. In this paper, we shall minimize the average transmission resource consumptions for BS multicasts, including the multicast power and the number of transmission symbols, by optimizing the scheduling policy for all file requests $ \{\Omega_{\tau}| \forall \tau\} $. 
	\begin{Remark}[Transmission Time of Each File Multicast]
		In this paper, we assume that each file multicast can be completed within the coherent time of shadowing effect, which is usually a few seconds. This is suitable for the transmission of short video clips. For the transmission of larger files, they can be divided into a number of segments, each of which can be delivered within the coherent time of shadowing effect. Hence, we can either treat each segment transmission as one new file request, or extend the scheduling policy for one file request in Definition \ref{def:policy} by adopting different scheduling parameters for different segments. The latter extension can follow the similar approach to our previous work \cite{Lv2019}.
	\end{Remark}
	\begin{Remark}[Overlapped Transmission Period]
		It is possible that one file request (say the ($\tau+1$)-th file request) arrives in the transmission period of the previous file request (the $\tau$-th request). It may not be necessary to postpone the ($\tau+1$)-th file transmission in this case. For example, the $N_\tau$ and $N_{\tau+1}$ transmission symbols for both file multicasts may share the same frames (in Fig. \ref{fig:frame}, one frame may carry symbols for both file multicasts). Hence, the transmission periods for the two file requests can overlap, as long as each file can be delivered within the coherent time of shadowing effect. In order to maintain the latest system state in such concurrent file transmission, the BS can update the CaSI from $B_\tau$ to $B_{\tau+1}$ immediately after the decision on $\{\Delta \mathcal B^c_{f,\tau} |\forall c \}$ is made. The postpone of the ($\tau+1$)-th file transmission is necessary when it relies on the previous file transmission. For example, if the ($\tau+1$)-th file transmission is from one cache node to the requesting user, and this cache node is receiving the same file from the BS on the $\tau$-th file request. It is clear that the chance of this situation is very small. 
	\end{Remark}

\section{Problem Formulation and Optimal Policy Structure} \label{sec:formulation}

\subsection{Problem Formulation}

In practice, the BS should deliver both popular and dedicated data in downlink as in Fig. \ref{fig:frame}. The former can be assisted by cache nodes; and the latter, e.g., video call is usually dedicated to one particular user, has to be served by the BS. In this paper,  we shall minimize the transmission resource consumption at the BS by offloading traffic for popular files to the cache nodes, so that more downlink transmission resource can be spared for the delivery of dedicated data.

Specifically, let $ \mathcal{C}_{f,\tau} \triangleq \bigcup\limits_{\{c|\forall\mathcal{B}_{f,\tau}^c=1\}} \mathcal{C}_c $ be the coverage area of the cache nodes which have already decoded the $f$-th file before the $\tau$-th file request, and $ \mathbf{l}_{\tau} $ be the location of the $\tau$-th user. The resource consumption at the $ \tau  $-th file transmission, measuring the weighted sum of the transmission energy and the number of transmission symbols, is given by $$g_\tau \Big(S_{\tau}, \Omega_{\tau}(S_{\tau})\Big) \triangleq(P_{\tau}N_{\tau} + w N_{\tau}) \times \mathbf{I}  [\mathbf{l}_{\tau} \notin \mathcal{C}_{\mathcal{A}_{\tau},\tau}],$$ 
where $ w$ is the weight on the number of transmission symbols. The average total cost of the BS in the whole lifetime is given by $$
\overline{G} (\{\Omega_1, \Omega_2, ...\}) \triangleq \mathbb{E}_{\mathcal{A},\zeta,{N}_R}\Big \{ \sum_{\tau=1}^{ {N}_R}  
g_{\tau} \Big(S_{\tau}, \Omega_{\tau}(S_{\tau})\Big) \Big\},$$
where $ \zeta \triangleq \{\zeta_\tau|\forall \tau=1,...,N_R\}$ and $ \mathcal{A}\triangleq\{\mathcal{A}_\tau|\forall \tau=1,...,N_R\}$. The expectation is taken over all possible large-scale channel fading, requested files and the total request number in the lifetime $N_R$. As a result, the transmission design in this paper can be formulated as the following stochastic optimization problem.
\begin{Problem}[Optimization with Random Number of Requests] \label{prob:main}
\begin{eqnarray}
	&\min\limits_{\{\Omega_1, \Omega_2, ...\}} & \overline{G}(\{\Omega_1, \Omega_2, ...\}) =\mathbb{E}_{\mathcal{A},\zeta,{N}_R}\bigg \{ \sum_{\tau=1}^{ {N}_R}  
	g_{\tau} \bigg(S_{\tau}, \Omega_{\tau}(S_{\tau})\bigg) \bigg\} \nonumber \\
	&s.t.& \mbox{Constraints in } (\ref{eqn:peak_power}),  (\ref{constrain:user})-(\ref{constrain:cache_action2}).\nonumber
\end{eqnarray}
\end{Problem}

		\begin{Remark}[Interpretation of Problem \ref{prob:main}]
			 Problem \ref{prob:main} is about the BS's scheduling of file multicast within a finite lifetime, where the file requests arrive randomly. The file transmissions are offloaded from the BS, if the requesting users can be served by some cache nodes. Otherwise, the BS should make sure the requesting users are able to decode the requested file from downlink multicast, as shown in the hard decoding constraint \eqref{constrain:user}. Meanwhile, the BS can also choose some cache nodes as the receivers of the downlink multicast. Hence, each file multicast updates the cached files at the cache nodes, which affects the probability of traffic offloading in the future. Different choices of receiving cache nodes in each downlink multicast may lead to different probability of future traffic offloading and different overall transmission cost (objective of Problem \ref{prob:main}). Problem \ref{prob:main} is to find the best selection policy of receiving cache nodes for all file multicasts, via the adaptation of multicast power and symbol number.
		\end{Remark}

\subsection{Structure of Optimal Scheduling Policy} \label{sec:optimal}

Since the number of stages $N_R$ in Problem \ref{prob:main} is random, the standard approach for finite-horizon MDPs with a fixed number of stages \cite{Bertsekas2000Dynamic} cannot be adopted here. An intuitive explanation is as follows: in conventional finite-horizon MDP, the policy at one stage is determined according to the system state and the optimized cost of the remaining stages; in Problem \ref{prob:main}, the number of remaining stages is random, and the optimized cost of the remaining stages cannot be derived by the backward induction. Hence, we first revise the definition of \emph{value function} of MDP such that the Bellman's equations can be extended to the case of a random stage number. Specifically, let $W_k(S_k)$ be the revised value function for the $k$-th file request and system state $S_k$, which measures the minimum average remaining cost from the $ k $-th request to the last ($ N_R $-th with $ N_R $ being a random variable) request given the current system state. That is, 
\begin{eqnarray} 
{\mbox{\bf Revised Value Function: }}W_k(S_k)\triangleq&\min\limits_{\{\Omega_k, \Omega_{k+1}, ...\}} &\mathbb{E}_{\mathcal{A},\zeta,N_R}\bigg\{\sum_{\tau=k}^{N_R}g_{\tau}(S_{\tau}, \Omega_{\tau})\mathbf{I}(k \leq N_R) \bigg| S_k \bigg\} \nonumber\\
 &s.t.& \mbox{Constraints in } (\ref{eqn:peak_power}),   (\ref{constrain:user})-(\ref{constrain:cache_action2}). \nonumber
\end{eqnarray}
Compared to the conventional definition of value function, the above definition of $W_k(S_k)$ involves an extra indicator function $ \mathbf{I}(k \leq N_R) $ and an expectation on $ N_R $, which count for the randomness of $ N_R $ and the situation that $ N_R $ may be smaller than $ k $. Then, the optimal solution of Problem \ref{prob:main} can be deduced via the following revised Bellman's equations.
 \begin{Lemma}[Revised Bellman's Equations for MDP with Random Number of Stages] \label{lem:MDP with random stage}
The value function $W_k(S_k), \forall S_k, k=1,2,\dots,L,$ satisfy the following Bellman's equations.
\begin{align}
	W_{k}(S_{k}) = &\min\limits_{\Omega_{k}(S_{k})} \bigg\{ g_{k}\bigg(S_{k}, \Omega_{k}(S_{k})\bigg)\Pr(N_R\geq k)
	+  \sum_{S_{k+1}} W_{k+1}(S_{k+1})\Pr\bigg[S_{k+1}\bigg|S_{k},\Omega_{k}(S_{k})\bigg]  \bigg \}, \nonumber\\
	&s.t. \quad \mbox{Constraints in } (\ref{eqn:peak_power}),  (\ref{constrain:user}) - (\ref{constrain:cache_action2})\nonumber,
\end{align}
 	where  $ {S}_{k + 1} $ denotes the system state at the $ (k+1) $-th request, and $W_{L+1}(.)\equiv0$ for notation convenience. 
\end{Lemma}
\begin{proof}
 	Please refer to Appendix A.
\end{proof}
 
The value functions $ \{W_k|\forall k\} $ are the functions of both CaSI and SCSI. As the space of the latter is continuous, the calculation of value functions for all system states is intractable. In this paper, by exploiting the independence between the distributions of the large-scale fading and requested files, and adopting the approach proposed in \cite{RuiWang2011}, we reduce the system state space. Note that the system state at the $ k $-th request can be represented by $ S_k = [\mathcal A_k, B_k, \zeta_k] $ and the distributions of $ \mathcal A_k $ and $ \zeta_k $ in each request are independent, we can obtain the following equivalent Bellman's equations with reduced state space by taking expectation over $ \mathcal A_k $ and $ \zeta_k $.
\begin{align}\label{eqn:bellman-reduce}
\widetilde{W}_{k}(B_{k}) \! =\!\! \min\limits_{\Omega_{k}(B_{k})} \! \!  \mathbb{E}_{\mathcal{A}_{k},\zeta_k} \bigg\{ \! {g}_{k}\!\bigg(\!\! S_{k}, \Omega_{k}(S_{k})\!\! \bigg)\! \Pr(N_R\!\geq\! k)
\!+\!\! \sum_{B_{k+1}} \!\! \widetilde{W}_{k+1}(B_{k+1})\!\Pr\!\bigg[\!B_{k+1}\bigg|{S}_{k},\Omega_{k}({S}_{k}) \!\bigg] \! \bigg \},
\end{align}
where  $\widetilde{W}_{k}(B_{k})\triangleq\mathbb{E}_{\mathcal{A}_{k},\zeta_k}\Big[W_{k}(S_{k})\Big|B_k\Big] $ is the new value function, and $ \Omega_{k}(B_{k}) \triangleq \{\Omega_{k}(S_{k})|\forall \mathcal{A}_{k},  \zeta_k \} $ is the aggregation of scheduling policy for all possible requested files, pathlosses and shadowing coefficients of the requesting users and the cache nodes. $ B_{k + 1} $ denotes the CaSI at the $ (k+1) $-th request. Moreover, we have the following conclusion.

\begin{Lemma}[Optimal Scheduling Policy]
The optimal scheduling policy for Problem \ref{prob:main} is
\begin{align}\label{eqn:policy}
&\Omega_{k}^{*}(S_k)\!=\!\arg\!	\min\limits_{\Omega_{k}(S_{k})} \! \bigg\{ \!{g}_{k}\bigg(\!\!S_{k}, \Omega_{k}(S_{k})\!\!\bigg)\!\Pr(N_R\geq k) 
\!+\!\! \sum_{B_{k+1}} \widetilde{W}_{k+1}(B_{k+1})\!\Pr\!\bigg[\!B_{k+1}\bigg|{S}_{k},\Omega_{k}({S}_{k})\!\bigg]  \! \bigg \}, 
\end{align}

\end{Lemma}
\begin{proof}
	The proof is similar to that of Lemma 1 of \cite{Lv2019} and is omitted due to page limitation.
\end{proof}
 In the standard optimal solution algorithm for finite-horizon MDP, we could start by evaluating the value function for the last stage (the $L$-th stage) via
\begin{align}
\widetilde{W}_{L}(B_{L})=\min\limits_{\Omega_{L}(B_{L})}\mathbb{E}_{\mathcal{A}_{L},\zeta_L}\bigg\{{g}_{L}\bigg(S_{L}, \Omega_{L}(S_{L})\bigg)\Pr(N_R\geq L) \bigg\},\nonumber
\end{align}
and then evaluate the value function of its previous stage ($\widetilde{W}_{L-1}(B_{L-1})$) via
\begin{align}
\widetilde{W}_{L-1}(B_{L-1})  =& \min\limits_{\Omega_{L-1}(B_{L-1})}  \mathbb{E}_{\mathcal{A}_{L-1},\zeta_{L-1}} \bigg\{  {g}_{L-1}\!\bigg( S_{L-1}, \Omega_{L-1}(S_{L-1}) \bigg) \Pr(N_R\geq L-1)\nonumber\\
&+ \sum_{B_{L}}  \widetilde{W}_{L}(B_{L})\Pr\bigg[B_{L}\bigg|{S}_{L-1},\Omega_{L-1}({S}_{L-1}) \!\bigg]  \bigg \}.\nonumber
\end{align}

By such backward induction from the last stage to the first stage, the value functions $\widetilde{W}_{k}(B_{k})$, $\forall k=1,...,L, $ can all be calculated, and the optimal scheduling policy can be derived via \eqref{eqn:policy}. However, this optimal algorithm usually suffers from the curse of dimensionality, which will be further explained in the following section.
	
\section{Low-Complexity Scheduling Policy}\label{sec:approximation}

The computation complexity of the value functions $ \widetilde{W}_k $ ($ k = 1,...,L$) defined in (\ref{eqn:bellman-reduce}) is huge due to the following two reasons. First, the maximum possible stage number $ L $ is usually huge, and the value functions for all stages should be evaluated, unlike the infinite-horizon MDP considered in most of the existing literature \cite{cui2010,Wang2013,AnLiu2014}. Consider one example with a lifetime of $ 24 $ hours and a frame duration of $ 10 $ milliseconds. The number of frames in the lifetime is $ L= 100 \times 60 \times 60 \times 24 \approx 8.6 \times 10^6$. Moreover, the space of $ B_k $ grows exponentially with respect to the number of cache nodes and the number of files. Note that the conventional approaches for approximate MDP are mostly designed for infinite-horizon MDP. In this section, we shall propose a novel framework to approximate the value functions of finite-horizon MDP (with a random stage number), such that the computation complexity of the value functions can be significantly reduced.

Specifically, regarding the above two causes for prohibitive complexity, the approximation of the value functions consists of the two steps.\footnote{ { Note that the average number of file requests $ \beta L $ is usually much smaller than the total number of frames $ L $ (there will be traffic overflow otherwise), it may be costly and inefficient to evaluate the value functions for all $ k=1,2,...,L $.}} Firstly, we propose a flexible approximation framework in Section \ref{sub:random-stage}, where the value functions are approximated by their lower-bounds, and the approximation error can be adjusted as a trade-off of the computation complexity. Secondly, the lower-bounds are further decoupled for each file and each cache node via a novel linear approximation method in Section \ref{sub:linear-app}. In addition, the overall approximation error is analyzed in Section \ref{Sub:Sec:bounds}, and an online scheduling algorithm based on the approximate value functions is elaborated in Section \ref{sub:online}.

\subsection{Flexible Approximation Framework for Stage Number Reduction}\label{sub:random-stage}

We first introduce the following bounds on the value functions $ \widetilde{W}_k $, $ k = 1,...,L$.

\begin{Lemma} \label{lem:reduce stage number}
	Let $ M_{R}^{\epsilon} \triangleq \max\{M_R^{\epsilon} \in \mathbb{Z}_{+}|\Pr (N_R > M_R^{\epsilon}) \leq \epsilon\} $. An upper-bound on $\widetilde{W}_{k}({B}_{k}) $ ($ \forall k\leq M_R^{\epsilon} $) is given by
	\begin{eqnarray}
	\widetilde{W}_{k}({B}_{k}) 
	\leq \min\limits_{\{\Omega_{\tau}|\forall \tau\}} \bigg\{\mathbb{E}_{\mathcal{A},\zeta}\bigg[ \sum_{\tau=k}^{M_R^{\epsilon}}  
	g_{\tau} ({S}_{\tau}, \Omega_{\tau}) \Pr(\tau \leq N_R)  \bigg| {B}_{k} \bigg]\bigg\}+ \!\!\!\!\!\sum_{\tau=M_R^{\epsilon}+1}^{L} \!\!\overline{g}_{max} \Pr (\tau \leq N_R) 
	 \triangleq  \mathcal{U}_k({B}_{k}),  \nonumber
	\end{eqnarray}
	where $\overline{g}_{max}\triangleq \mathbb{E}_{\eta_{\tau}}  \Big[ 	F(\theta,P_B) \Big| \rho_{\tau}=\rho_{min} \Big]$ denotes an upper-bound on the average cost of each stage, $F(\theta,P_B)$ is given by
		\begin{align}\label{eqn:F}
	F(\theta,P_B)\triangleq \begin{cases}
	\frac{w\ln(2)R_F} {\alpha \mathbb{W}(\frac{2^{\theta}w}{e})},&	\frac{w} { \mathbb{W}(\frac{2^{\theta}w}{e})}<P_B\cr
	\frac{(P_B+w)R_F}{\log_2(P_B)+\theta},&\frac{w} {\mathbb{W}(\frac{2^{\theta}w}{e})}\geq P_B
	\end{cases}, \quad \forall \theta,
	\end{align}
	$\theta \triangleq	\mathbb{E}_{\mathbf{h}_{\tau}} \left[ \log_2 \left(  \frac{||\mathbf{h}_{\tau}||^2}{N_T \sigma^2_z}\right) \right] $, and $\mathbb{W}(x)$ is the Lambert-W function \cite{W}, $\rho_{min}$ is the pathloss from the BS to the farthest location of the cell. Moreover, a lower-bound on  $\widetilde{W}_{k}({B}_{k}) $  ($ \forall k\leq M_R^{\epsilon} $) is given by
	 	\begin{align} \label{eqn:lower-bound:reduce_system_stage}
		\widetilde{W}_{k}({B}_{k}) \geq \min\limits_{\{\Omega_{\tau}|\forall \tau\}} \mathbb{E}_{\mathcal{A},\zeta}\bigg \{ \sum_{\tau=k}^{M_R^{\epsilon}}  
		g_{\tau} ({S}_{\tau}, \Omega^{*}_{\tau}) \Pr (\tau \leq N_R) \bigg| {B}_{k} \bigg\} \triangleq  \mathcal{L}_k({B}_{k}).
	\end{align}
	\end{Lemma}
\begin{proof}
	Please refer to Appendix B.
\end{proof}

As a remark, notice that $ \overline{g}_{max}$ is actually the minimum average transmission cost to one requesting user located at the farthest position from the BS. It is clear that when increasing $ M_R^{\epsilon} $, the lower-bound $ \mathcal{L}_k $ and the upper-bound $ \mathcal{U}_k $ tend to the exact function $ \widetilde{W}_k $, { i.e. the approximation error tends to zero}, at the cost of computation complexity and storage increase. The approximation error, in terms of $ \epsilon $, will be discussed in Section \ref{Sub:Sec:bounds}. In this paper, we shall use the lower-bound to approximate the value functions when $ k \leq M_R^{\epsilon} $, and zero to value functions when $ k>M_R^{\epsilon} $, i.e.,
\begin{equation}\label{eqn:approx-stage}
\widetilde{W}_k( B_k ) \approx \left\{ \begin{array}{cc}
	\mathcal{L}_k(B_{k}) & \ k \leq M_R^{\epsilon} \\ 
	0 & \ k > M_R^{\epsilon}
	\end{array} \right. 
\end{equation}
Direct evaluation of $ \mathcal{L}_k $ for $ k=1,2,...,M_R^{\epsilon} $ is still intractable due to the huge space of $ B_k $. In the following part, we shall further decouple $ \mathcal{L}_k $ via a novel linear approximation method.

\subsection{Linear Approximation of Value Function}\label{sub:linear-app}
Let $\{\Omega^{\dagger}_\tau | \tau=1,...,M_R^{\epsilon}\} = \arg\min_{\{\Omega_{\tau} | \tau=1,...,M_R^{\epsilon}\}} \mathbb{E}_{\mathcal{A},\zeta}\bigg \{ \sum_{\tau=1}^{M_R^{\epsilon}}  
g_{\tau} (S_{\tau}, \Omega_{\tau}) \Pr (\tau \leq N_R) \bigg\} $
be the optimal scheduling policy with fixed $ M_R^{\epsilon} $ stages. The value function for $ k=1,...,M_R^{\epsilon} $ and CaSI $ B_{k} $  can be written as 
\begin{eqnarray}
\widetilde{W}_{k}({B}_{k}) 
&\approx&  \mathcal{L}_k ({B}_{k}) = \sum_{f=1}^{M_F} \underbrace{ \sum_{c=0}^{N_C} \overbrace{ \mathbb{E}_{\mathcal{A},\zeta}\bigg[ \sum_{\tau=k}^{M_R^{\epsilon}}  
		g_{\tau} ({S}_{\tau}, \Omega^{\dagger}_{\tau}) \Pr(\tau \leq N_R) \mathbf{I}(\mathcal{A}_{\tau} = f)\mathbf{I} (\mathbf{l}_{\tau} \in \mathcal{C}_c)  \bigg|  {B}_{k} }^{\widetilde{W}_{k,f,c}({B}_{k})}  \bigg] }_{\widetilde{W}_{k,f}({B}_{k})} \nonumber\\
&=& \sum_{f=1}^{M_F} \widetilde{W}_{k,f}({B}_{k}) = \sum_{f=1}^{M_F} \sum_{c=0}^{N_C} \widetilde{W}_{k,f,c}({B}_{k}),
\end{eqnarray}
where $\widetilde{W}_{k,f}({B}_{k})$ is named as the {\em per-file value function}, approximating the average cost for the $f$-th file since the $ k $-th stage under the optimal policy $ \{\Omega^{*}_k|\forall k \} $ defined in (\ref{eqn:policy}). $\widetilde{W}_{k,f,c}({B}_{k})$ is named as the {\em per-file per-region value function}, approximating the average cost for the $f$-th file in the region $\mathcal{C}_c$ since the $ k $-th stage under the optimal policy $ \{\Omega^{*}_k|\forall k \} $. The ways to approximate value function, as introduced in \cite{Shewhart2011Approximate}, are general but short at exploiting the problem structure. Moreover, they usually require value iteration and may not be applicable for finite-horizon MDP. In the following, we shall propose a problem-specific approximation of value functions with analytical expressions. With such analytical expressions, computationally expensive numerical algorithms such as value iteration can be bypassed.

We refer to $  \mathcal{F}_{c}^H \triangleq \{1,2,...,M_c\},c=1,2,\cdots,N_C,$ and $  \mathcal{F}_{c}^L \triangleq \{M_c+1, ..., M_F\} $ as the {\em high-popularity} and {\em low-popularity} file sets in the region $\mathcal{C}_c$, respectively. If the $ f $-th file ($ f\in \mathcal{F}_c^H $) has been stored, it will never be replaced by other files. On the other hand, it is possible that $f$-th file ($ f\in \mathcal{F}_c^L $) stored at the $c$-th cache node will be replaced by other files in the future transmission due to limited cache space. The approximations of the per-file per-region value functions $ \widetilde{W}_{k,f,c} $, $\forall k, f, c$, which are denoted as $ J_{k,f,c} $, are elaborated below respectively.

\subsubsection{Approximation of $  \widetilde{W}_{k,f,0}  $} 

$ \widetilde{W}_{k,f,0} $ is the transmission cost spent  for any requesting users located in the region without cache nodes $ \mathcal{C}_0 $. In the approximation, we assume that the BS spends just sufficient transmission cost to ensure the file delivery to any requesting users in $ \mathcal{C}_0 $. Hence, the approximation of $ \widetilde{W}_{k,f,0} $, denoted as $J_{k,f,0}$, can be written as 
	\begin{align} 
	J_{k,f,0}
	\triangleq\min_{\{P_{\tau},N_{\tau}|\tau=k,...,M_R^{\epsilon}\}} \ \sum_{\tau=k}^{M_R^{\epsilon}}
	p_f\Pr[N_R\geq\tau] \Pr[\mathbf{l}_{ \tau } \in \mathcal{C}_{0}] \times\mathbb{E}_{\zeta_{\tau}}  \bigg[  P_{\tau}N_{\tau}+wN_{\tau} \bigg|\mathbf{l}_{\tau}\in \mathcal{C}_{0} \bigg] \label{eqn:J_{k,f,0}}, 
	\end{align}
	where the constraints in \eqref{eqn:peak_power} and \eqref{constrain:user} should be satisfied.
	It is clear that\footnote{Please refer to \cite{Lv2019} for the derivation of the expression.}
	\begin{align}
	J_{k,f,0} \approx \sum_{\tau=k}^{M_R^{\epsilon}}
	p_f\Pr[N_R\geq\tau]\underbrace{ \Pr[\mathbf{l}_{ \tau } \in \mathcal{C}_{0}] \times\mathbb{E}_{\zeta_{\tau}}  \bigg[F(\theta,P_B) \bigg| \mathbf{l}_{\tau}\in \mathcal{C}_{0} \bigg]}_{\mbox{Denoted as } \mu_{0}}, \label{eqn:app_J_{k,f,0}}
	\end{align}
	where the approximation is for high SINR region,  $F(\theta,P_B)$ is given by \eqref{eqn:F} and $\\\theta =	\mathbb{E}_{\mathbf{h}_{\tau}} \left[ \log_2 \left(  \frac{||\mathbf{h}_{\tau}||^2}{N_T \sigma^2_z}\right) \right] $. If the statistics of large-scale fading and {file popularity} are known, the above expectation can be calculated. Otherwise, a learning-based approach is introduced in the next section to evaluate the approximate value function.

\subsubsection{Value Function Approximation for High-Popularity Files}
{ Let ${\mathbf{b}}_{f}^c\triangleq \{\mathcal{B}_{f,k}^{c} =0\}\cup \{\mathcal{B}_{f,k}^{i} =1| \forall i\neq c\}$ be the CaSI where only the $ c $-th cache node has not decoded the $ f $-th file.  We first define $ \widehat{W}_{k,f} (\mathbf{b}_{f}^c)$ as follows.
	\begin{eqnarray}\label{eqn:w-lowerbound}
	\widehat{W}_{k,f} (\mathbf{b}_{f}^c) & \triangleq & \min \mathbb{E}_{\mathcal{A},\zeta} \bigg\{ \sum_{\tau=k}^{M_R^{\epsilon}}  
	g_{\tau} \mathbf{I}(\mathcal{A}_{\tau} = f) \Pr (\tau \leq N_R) \bigg| \mathbf{b}_{f}^c \bigg\} \\
	&s.t.& \mbox{Constraints in } (\ref{eqn:peak_power}),(\ref{constrain:user}) - (\ref{constrain:cache_action2}); \quad \Delta \mathcal{B}_{f,\tau}^i=0, \forall i \neq c ,\tau \geq k .\nonumber
	\end{eqnarray}
	$ \widehat{W}_{k,f} (\mathbf{b}_{f}^c) $ approximates the transmission cost for the $ f $-th file when the $ c $-th the cache node has not decoded the $ f $-th file, i.e., $ \widetilde{W}_{k,f,c} + \widetilde{W}_{k,f,0}$. Since $ J_{k,f,0} $ approximates the cost $ \widetilde{W}_{k,f,0} $, the per-file per-region value function $\widetilde{W}_{k,f,c}({B}_{k})$ for all $c>0$ and $f\in \mathcal{F}^H_c$ can be approximated by 
	\begin{eqnarray}\label{eqn:approx-high}
	\widetilde{W}_{k,f,c}({B}_{k}) \approx J_{k,f,c}({B}_{k})\triangleq
	\begin{cases}
	0, &\mbox{when } \mathcal{B}_{f,k}^c=1; 
	\cr \widehat{W}_{k,f}(  {\mathbf{b}}_{f}^c )-J_{k,f,0}, &\mbox{when } \mathcal{B}_{f,k}^c=0.
	\end{cases}
	\end{eqnarray}
	Moreover, $\widehat{W}_{k,f} (\mathbf{b}_{f}^c)$ can be calculated via the following backward induction:
	\begin{itemize} 
		\item {\bf{Step 1}}: Let $k=M_R^{\epsilon}+1$, and initialize $\widehat{W}_{k,f}(\mathbf{b}_{f}^c)=0$.
		
		\item {\bf{Step 2}}: Let $k=k-1$, and calculate $\widehat{W}_{k,f}( \mathbf{b}_{f}^c )$ as follows.
		\begin{align}\label{eqn:v_default one fragment }
		\widehat{W}_{k,f}( \mathbf{b}_{f}^c )&= (1-p_f) \times\widehat{W}_{k+1,f}( \mathbf{b}_{f}^c )  
		+p_f \times \upsilon_{k,f,c} 
		\end{align}
		\begin{align}
		\upsilon_{k,f,c}&= \Pr\Big[\mathbf{l}_{k} \in \mathcal{C}_{f,k}( \mathbf{b}_{f}^c )\Big] \times {\widehat{W}}_{k+1,f}( \mathbf{b}_{f}^c )\nonumber\\
		&+\Pr\Big[\mathbf{l}_{k} \notin \mathcal{C}_{f,k}( \mathbf{b}_{f}^c )\Big]\times\Big\{ \mathbb{E}_{\zeta_k}\Big[\overline{G}_{k,f}^1\Big|R_k \leq R_k^c   \Big]\Pr\Big[R_k \leq R_k^c\Big] \nonumber \\  
		&+\mathbb{E}_{\zeta_k}\Big[\min \{\overline{G}_{k,f}^2, \overline{G}_{k,f}^3\}\Big|R_k > R_k^c  \Big]\Pr\Big[R_k > R_k^c  \Big] \Big\} ,\label{eqn:v_{k,f,c}}
		\end{align}
		where $\mathcal{C}_{f,k}$ is the area (given CaSI $\mathbf{b}_{f}^c$) in which the users are able to receive the $ f $-th file from one of the cache nodes, and $\overline{G}_{k,f}^1,\overline{G}_{k,f}^2,\overline{G}_{k,f}^3$ are given by 
		\begin{align} \label{eqn:Q}
		&\overline{G}_{k,f}^1 \triangleq 	  F(\theta,P_B)\Pr[N_R \geq k] +{ {J}_{k+1,f,0}},
		\overline{G}_{k,f}^2\triangleq	 F(\theta,P_B)\Pr[N_R \geq k] +{ {\widehat{W}}_{k+1,f}(          \mathbf{b}_{f}^{c})}, \nonumber \\   
		&\overline{G}_{k,f}^3\triangleq 	  F(\theta^c,P_B) \Pr[N_R \geq k] +{ {J}_{k+1,f,0}}, \text{where } \theta^c \triangleq	\mathbb{E}_{\mathbf{h}^c_{k}} \left[ \log_2 \left(  \frac{||\mathbf{h}_{k}^c||^2}{N_T \sigma^2_z}\right) \right].
		\end{align}	
	
		\item {\bf{Step 3}}: If $k>1$, go to step 2. Otherwise, terminate. 
	\end{itemize}
}

\subsubsection{Value Function Approximation for Low-Popularity Files}
In order to approximate $ \widetilde{W}_{k,f,c} $, $ f \in  \mathcal{F}_{c}^L$, we first define the following notations.
\begin{itemize}
	\item Let $ \eta_{k,f,c}$ be the minimum transmission cost for the BS to { ensure successful file transmission when the requesting user is in the region} $ \mathcal{C}_c $ since the $\tau$-th frame, given that the $ f $-th file has not been cached at the $ c $-th cache node, i.e.,
$$
	\eta_{k,f,c}
	\triangleq\min_{\{P_{k},N_{k}\}} \ 
	p_f\Pr[N_R\geq k] \\\times\Pr[\mathbf{l}_{ k } \in \mathcal{C}_{c}] \times\mathbb{E}_{\zeta_k}  \Big[  P_{k}N_{k}+wN_{k} \Big|\mathbf{l}_{k}\in \mathcal{C}_{c} \Big], \nonumber 
$$
	where the constraints in \eqref{eqn:peak_power} and \eqref{constrain:user} should be satisfied.	It is clear that
		$$\eta_{k,f,c}{\approx} p_f \Pr[N_R\geq\tau]\Pr[\mathbf{l}_{k}\in \mathcal{C}_{c} ]\mathbb{E}_{\zeta_k}\Big[ F(\theta,P_B) \Big|\mathbf{l}_{k}\in \mathcal{C}_{c} \Big]$$ for the high SINR region.
	
	\item Let $ \mathbb{P}^{f,c}_{k,n}({B}_{k}) $ be the probability that there are $ M_c - \sum_{m=1}^{f-1} \mathcal{B}_m^c -1$ requests for the files $ \{1,2,...,f-1\} $ from the $ k $-th request to the $ (k+n-1) $-th request, and  the $ (k+n) $-th request is also for these files, given the current CaSI $ {B}_{k} $. Hence,
	\begin{eqnarray}
	{\mathbb{P}}^{f,c}_{k,n}(\!{B}_{k}\!)\!\triangleq\! \begin{cases} 0, &n\!< M_c \!-\! \sum_{m=1}^{f-1} \mathcal{B}_m^c \\ \binom{n-1}{M_c - \sum_{m=1}^{f-1} \mathcal{B}_m^c -1}	\varphi^{M_c - \sum_{m=1}^{f-1} \mathcal{B}_m^c }(1-\varphi)^{n-M_c + \sum_{m=1}^{m-1} \mathcal{B}_{m}^{c} }, &n\!\geq M_c \!-\! \sum_{m=1}^{f-1} \mathcal{B}_m^c 
	\end{cases}
	\end{eqnarray}
	where $\varphi\triangleq\sum_{\{f|f\in \mathcal{F}^H_c, \mathcal{B}_f^c=0\} } p_f$.
\end{itemize}

For $ f \in  \mathcal{F}^L_c$, { the per-file per-region value function is approximated as follows.}
\begin{eqnarray}\label{eqn:low-pop-app}
\widetilde{W}_{k,f,c}({B}_{k}) \approx J_{k,f,c}({B}_{k})\triangleq
\begin{cases}
\sum\limits_{\tau = k}^{M_R^{\epsilon}} \eta_{\tau,f,c} - \sum\limits_{n=0}^{M_R^{\epsilon}-k}\sum\limits_{\tau=k}^{k+n} \eta_{\tau,f,c} {\mathbb{P}}^{f,c}_{k,n}({B}_{k}), & \mbox{when } \mathcal{B}_{f,k}^c=1 \\
\sum\limits_{\tau = k}^{M_R^{\epsilon}} \eta_{\tau,f,c} , &\mbox{when } \mathcal{B}_{f,k}^c=0
\end{cases}
\end{eqnarray}

\begin{figure}[tb]
	\centering
	\includegraphics[height=220pt,width=400pt]{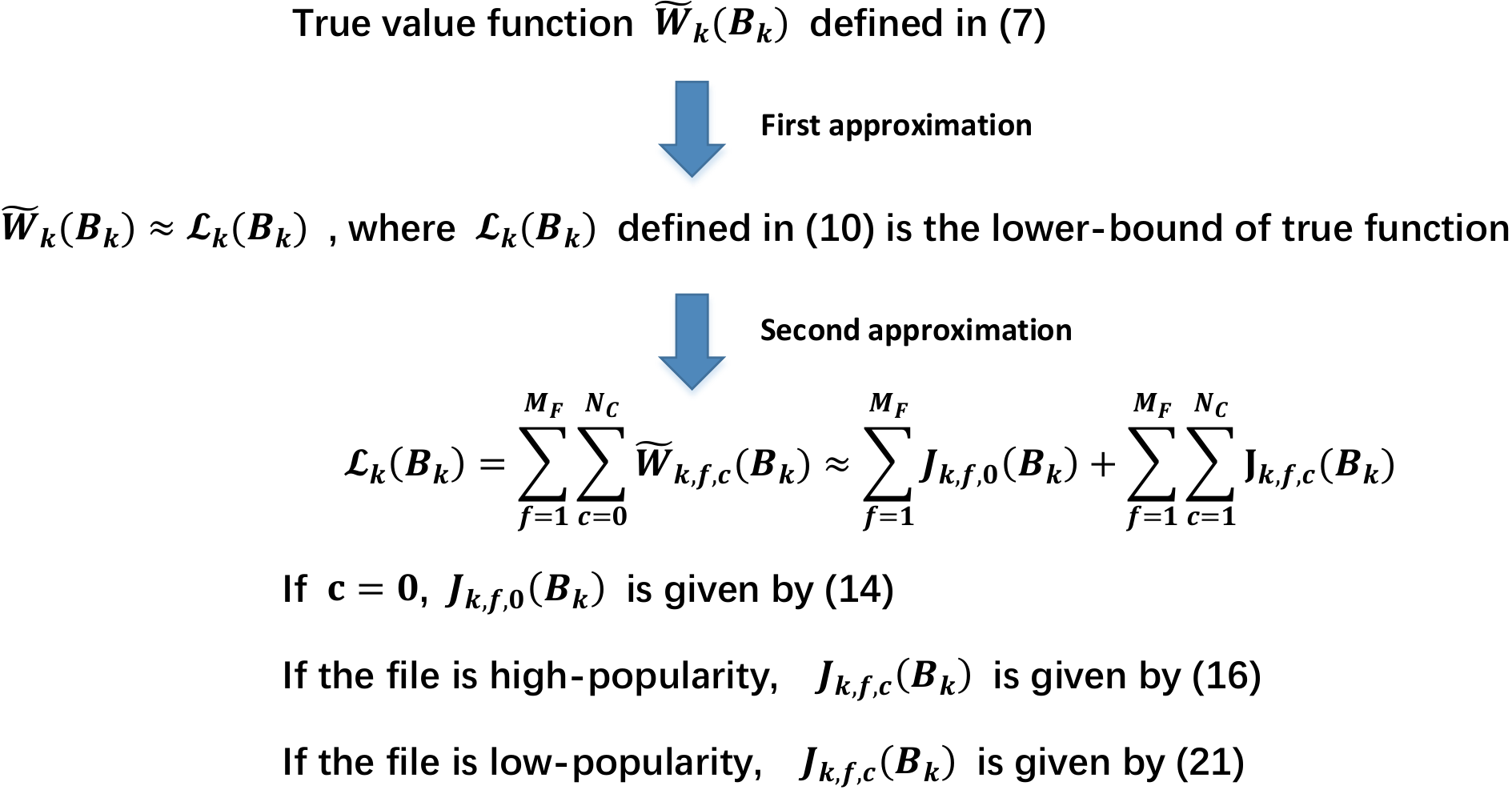}
	\caption{Illustration of overall value function approximation procedure.}
	\label{fig:appr_value_process}
\end{figure}

In summary, the overall value function approximation procedure is illustrated in Fig. \ref{fig:appr_value_process}. Compare with the conventional optimal solution, the complexity of value function evaluation is dramatically reduced. In the above approximation approach, the computation complexity of value function calculation is $\mathcal{O}(M_R^{\epsilon}N_CM_F)$. In order to store these values, the required memory space is also $\mathcal{O}(M_R^{\epsilon}N_CM_F)$. On the other hand, the optimal solution of MDP suffers from the curse of dimensionality. Specifically, the computation complexity of the conventional value iteration algorithm is $\mathcal{O}(LN_C2^{M_FN_C})$, and the memory requirement is  $\mathcal{O}(L2^{M_FN_C})$.

\subsection{Bounds on Approximate Value Functions}\label{Sub:Sec:bounds}

As elaborated in the above two parts, the approximation of value function $ \widetilde{W}_k $ is made via two steps. Firstly, the random stage number $ N_R $ is transformed to a fixed stage number $ M_R^{\epsilon} $, and the value function $ \widetilde{W}_k $ is approximated by its lower-bound $ \mathcal{L}_k $. Then, $ \mathcal{L}_k $ is further approximated by $ J_k$, i.e  $$J_{k}(B_k)\triangleq \sum_{f=1}^{M_F} J_{k,f,0}({B}_{k})+\sum_{c=1}^{N_C}  \Big\{\sum_{f=1}^{M_F} J_{k,f,c}({B}_{k}) \Big \}\approx \widetilde{W}_k(B_k),\forall B_k.$$ From Lemma \ref{lem:reduce stage number}, it is straightforward that the approximation error of the first step is upper-bounded as follows.
\begin{equation}
\widetilde{W}_k ({B}_k) - \mathcal{L}_k ({B}_k) \leq \mathcal{U}_k ({B}_k) - \mathcal{L}_k ({B}_k) = \sum_{\tau=M_R^{\epsilon}+1}^{L}	\overline{g}_{max} \Pr (\tau \leq N_R) < \epsilon L  \overline{g}_{max}.
\end{equation}
The error upper-bound $ \epsilon L \overline{g}_{max} $ tends to $ 0 $ when $ M_R^{\epsilon} $ tends to $ L $. Since $ \beta \ll 1 $, $ \epsilon L $ can be arbitrarily small even when $ \beta L < M_R^{\epsilon} \ll L $. As $ \overline{g}_{max} $ is a constant, we can choose one $ M_R^{\epsilon} \ll L $ such that the the upper bound $  \epsilon L  \overline{g}_{max} $ of $|\widetilde{W}_k ({B}_k) - \mathcal{L}_k ({B}_k)|$ is negligible compared with $ \mathcal{L}_k ({B}_k)$. Thus, the first approximation step of the value function can be tight. In order to analyze the overall approximation error (i.e. $|\widetilde{W}_k ({B}_k) - {J}_k ({B}_k)|$), we first introduce the following bounds.

\begin{Lemma}[Upper-Bound on Value Function] \label{lem:upper-bound}
	The value function $ \widetilde{W}_{k}({B}_{k}) $ is upper-bounded as
	\begin{equation}
	\widetilde{W}_{k} ({B}_{k}) \leq \mathcal{U}_{k} ({B}_{k}) \leq J_{k}({B}_{k}) + \sum_{\tau=M_R^{\epsilon}+1}^{L}	\overline{g}_{max} \Pr (\tau \leq N_R)  \triangleq \widetilde{\mathcal{U}}_k({B}_{k}), \ \ \forall k.
	\end{equation}
\end{Lemma}

\begin{proof}
	Please refer to Appendix C.
\end{proof}

Let $ {B}_{\tau+1}^{\pi} \triangleq \pi({B}_{\tau}) $ be the CaSI of the $ (\tau+1) $-th stage after applying the operator $  \pi $ on the CaSI $ {B}_{\tau} $ ($ \forall \tau $). Specifically, the cache status at the $ c $-th cache node ($ \forall c $) in $ {B}_{\tau+1}^{\pi} $ is given by
\begin{itemize}
	\item When $\sum_f \mathcal{B}_{f,\tau}^c<M_c$, 
	\begin{eqnarray}
	\mathcal{B}_{f,\tau+1}^c =  \begin{cases}
	1, \ \ &f = f_{min}^c \\
	\mathcal{B}_{f,\tau}^c, \ \ &f \neq f_{min}^c
	\end{cases}  
	\end{eqnarray}
	where $f_{min}^c \triangleq \min \{f|\forall \mathcal{B}_{f,\tau}^c=0\}$ is the index of the most popular file which has not been cached at the $ c $-th cache node.
	\item When $\sum_f \mathcal{B}_{f,\tau}^c = M_c$ and $f_{min}^c < f_{max}^c $, where $f_{max}^c \triangleq \max \{f|\forall \mathcal{B}_{f,\tau}^c=1\}$ is the index of the least popular file cached at the $ c $-th cache node, 
	\begin{eqnarray}
	\mathcal{B}_{f,\tau+1}^c =  \begin{cases}
	1, \ \ &f = f_{min}^c \\
	\mathcal{B}_{f,\tau}^c, \ \ &f \notin \{f_{min}^c, f_{max}^c \} \\
	0, \ \ &f = f_{max}^c \\
	\end{cases}  
	\end{eqnarray}
	\item When $\sum_f \mathcal{B}_{f,\tau}^c = M_c$ and $f_{min}^c > f_{max}^c $, $\mathcal{B}_{f,\tau+1}^c =\mathcal{B}_{f,\tau}^c $, $\forall f$.
\end{itemize} 
Moreover, let $ \pi^k({B}_{\tau}) $ be the CaSI of the $ (\tau+k) $-th stage after applying the operator $  \pi $ on the CaSI of the $ \tau $-th stage $ {B}_{\tau} $ for $ k $ times ($ \forall \tau $). A lower-bound of value function is described below.

\begin{Lemma}[First Lower-Bound on Value Function]\label{lem:first lower-bound of value function} One lower-bound of $\widetilde{W}_{k}({B}_{k})$ is given by
		\begin{align}
		\widetilde{W}_{k} ({B}_{k}) \geq \mathcal{L}_k ({B}_{k}) \geq \sum_{\tau=k}^{M_R^{\epsilon}} \bar{g}_{min}(\pi^{\tau - k} ({B}_{k})) \Pr(\tau \leq N_R)\triangleq \widetilde{\mathcal{L}}_{k}^{1}({B}_{k}),
		\end{align}
		where 
		\begin{eqnarray}
	 \bar{g}_{min}({B}_{\tau})\triangleq&\min_{\Omega_\tau}&\mathbb{E}_{\mathcal{A}_{\tau},\zeta_{\tau}}\{ g_{\tau}({S}_{\tau},\Omega_\tau)|B_{\tau} \} \nonumber\\
		 &s.t.& \mbox{Constraint in }(\ref{eqn:peak_power}), (\ref{constrain:user}),\nonumber
		 \end{eqnarray}
		   denotes the minimum transmission cost for the BS to guarantee the successful reception of the requesting user.
\end{Lemma}
\begin{proof}
	Please refer to Appendix D.
\end{proof}

The above lower-bound may be loose when most of the files are of high popularity, as it underestimates the cost of file placement. By relaxing the cache size constraint in \eqref{constrain:cache}, we have another lower-bound as follows.

\begin{Lemma}[Second Lower-Bound on Value Function]\label{lem:second lower-bound of value function}
	$\widetilde{W}_{k}({B}_{k})$ is lower-bounded by
	\begin{align}
	\widetilde{W}_{k} ({B}_{k}) \geq \!\sum_{f=1}^{M_F}\sum_{n=1}^{L-k+1}\!\!\!\binom{ L-k+1}{n}\!(\beta p_f)^n (1-\beta p_f)^{ L-n-k+1}\!\bigg(\! \bar{g}_{min}^{f}({B}_{k})+(n-1)\mu_{0} \!\bigg)\triangleq \widetilde{\mathcal{L}}_{k}^{2}({B}_{k}),\nonumber
	\end{align}
	where $\mu_0$ is given by \eqref{eqn:app_J_{k,f,0}}, and $\bar{g}_{min}^{f}({B}_{k})\triangleq\mathbb{E}_{\zeta_k}  \Big[  F(\theta,P_B)\Big| {B}_{k}, \mathbf{l}_{k} \in \mathcal{C} - \mathcal{C}_{f,k} \Big]\Pr(\mathbf{l}_{k} \in \mathcal{C} - \mathcal{C}_{f,k}).$
\end{Lemma}
\begin{proof}
	This lower-bound is obtained by assuming that there is sufficient memory space in each cache node, so that all the files can be stored without being replaced. Hence, the conclusion of this lemma directly follows Lemma 5 in \cite{WANG2017}.
\end{proof}

Hence, a tighter lower-bound of $\widetilde{W}_{k}({B}_{k})$ is given by 
\begin{align}\label{eqn:tighter_lower_bound}
\widetilde{\mathcal{L}}_{k}({B}_{k})\triangleq\max\bigg(\widetilde{\mathcal{L}}_{k}^{1}({B}_{k}),\widetilde{\mathcal{L}}_{k}^{2}({B}_{k})\bigg).
\end{align}
Note that $ \{J_k | \forall k \} $ are the proposed approximation of value functions, and the upper-bounds of their approximation errors can be calculated via
$$|J_{k}({B}_{k}) - \widetilde{W}_{k}({B}_{k}) |
\leq \widetilde{\mathcal{U}}_{k} ({B}_{k}) - \widetilde{\mathcal{L}}_{k} ({B}_{k}), \forall k, {B}_{k},$$ 
where both $ \widetilde{\mathcal{U}}_{k} ({B}_{k}) $ and $ \widetilde{\mathcal{L}}_{k} ({B}_{k}) $ can be calculated analytically. According to the definition of the value function, the optimal average cost of the whole lifetime with initial CaSI $ {B}_1 $ is $ \widetilde{W}_1({B}_1) $, which can be bounded as
	$\widetilde{\mathcal{L}}_{1}({B}_{1})  \leq \widetilde{W}_1({B}_1) \leq \widetilde{\mathcal{U}}_{1} ({B}_{1})$.

\subsection{Scheduling Policy with Approximate Value Functions} \label{sub:online}
In this part, we optimize the multicast power $ P_k $, symbol number $ N_k $ and the decisions on cache update $ \{\Delta \mathcal B^c_{f,k}|\forall f,c\} $ for the $ k $-th ($ \forall k $) multicast, when the $ k $-th file request cannot be served by any cache node (i.e., $ \mathbf l_k \notin \mathcal C_{\mathcal A_k, k} $). Given $ P_k $, $ N_k $ and the system state, the cache nodes which can successfully decode the $ \mathcal A_k $-th file are determined, so are $ \{\Delta \mathcal B^c_{f,k}|\forall f,c\} $. Hence, with the approximate value functions $ \{J_{k,f,c}| \forall k,f,c\} $,  the optimization problem in (\ref{eqn:policy}) can be rewritten as follows.
\begin{Problem}[Scheduling with Approximate Value Functions]\label{prob:online}
\begin{eqnarray}
Q^{*}_k \triangleq &\min\limits_{P_k, N_k} & \! \bigg\{ \!{g}_{k}\bigg[S_{k}, \Omega_{k}(S_{k})\bigg]\!\Pr(N_R\geq k) 
\!+\! J_{k+1}(B_{k+1}) \bigg \} \nonumber\\
&s.t.&\mbox{Constraints in }  (\ref{eqn:peak_power}), (\ref{constrain:user})-(\ref{constrain:cache_action2}),
\end{eqnarray}
where $ B_{k+1} $ represents the CaSI after the $ k $-th multicast. Specifically, in the cache update from $ B_k $ to $ B_{k+1} $, a cache node stores the $ \mathcal A_k $-th file when (1) it has decoded this file, and (2)  there is spare memory or cached files with lower popularity (than $ \mathcal{A}_{k} $).
\end{Problem}
Because of the factor $ J_{k+1}(B_{k+1}) $ in the objective, Problem \ref{prob:online} is a mixed continuous and discrete optimization problem: the cache nodes for receiving the $ \mathcal A_{k} $-th file should be selected;  for given the receiving cache nodes, the transmission power and symbol number should be optimized accordingly. Its optimal solution is summarized below.

\begin{Lemma}[Optimal Solution of Problem \ref{prob:online}]\label{lem:online-scheduling}
	Let $ d_1,d_2,..., d_{N_C}$ be the indexes of the cache nodes, whose large-scale fading coefficients satisfy $ \rho^{d_1}\eta_k^{d_1} \leq  \rho^{d_2}\eta_k^{d_2} \leq ... \leq \rho^{d_{N_C}}\eta_k^{d_{N_C}}$ and $ \rho^{d_m}\eta_k^{d_m}\!\! \leq \rho_k \eta_k\leq \rho^{d_{m+1}}\eta_k^{d_{m+1}} $. Define  $J_{k}^c(B_k)\triangleq \sum_{f=1}^{M_F} J_{k,f,c}({B}_{k}) $,
	\begin{eqnarray}
	Q_{d_i}^{*}&
	\triangleq&\min \bigg\{(P_{k}N_{k} + w N_{k})\Pr(N_R \geq k) + \sum\limits_{c=d_i}^{d_{N_C}}   J_{k+1}^c(  {B}_{k+1} )  + \sum\limits_{c=d_1}^{d_{i-1}} J_{k+1}^c(  {B}_{k} )  \bigg\}\nonumber\\
	&s.t.& \mbox{Constraints in } (\ref{eqn:peak_power}), (\ref{constrain:user}) - (\ref{constrain:cache_action2});\quad R^{d_i}_{k} \geq R_F \nonumber
	\end{eqnarray}
	The minimized objective of Problem \ref{prob:online} is given by 
	\begin{equation}
	Q_k^{*} = \min_{i=1,2,...,m+1} Q_{d_i}^{*}. \label{eqn:opt-qk}
	\end{equation}
	Moreover, $[P^{d^{*}}_{k},N^{d^{*}}_{k}]=\Big[\min(\frac{w} { \mathbb{W}(\frac{2^{\theta^{d^{*}}}w}{e}) },P_B),\max(\frac{R_F\ln(2)}{\alpha[\mathbb{W}(\frac{2^{\theta^{d^{*}}}w}{e})+1]},\frac{R_F}{\alpha[\theta^{d_{*}}+\log_2(P_B)]})\Big],$ where $\\
	d^{*}=\arg \min\limits_{d_i} Q_{d_i}^* $ and $\theta^{d^{*}} =	\mathbb{E}_{\mathbf{h}^{d^{*}}_{k}} \left[ \log_2 \left(  \frac{||\mathbf{h}_{k}^{d^{*}}||^2}{N_T \sigma^2_z} \right) \right] $, is the asymptotically optimal solution of Problem \ref{prob:online} in high SINR region.
\end{Lemma}	

\begin{proof}
	The conclusion of (\ref{eqn:opt-qk}) is straightforward by noticing that $ Q_{d_i}^{*} $ is the optimal value when the $ d_i $-th cache node and the cache nodes with better downlink channel than the $ d_i $-th one can decode the multicasted file. Moreover, the expressions of $ P^{d^{*}}_{k} $ and $N^{d^{*}}_{k} $ can be derived similarly to (\ref{eqn:app_J_{k,f,0}}).
\end{proof}

It is clear from the above lemma that, the asymptotically optimal solution of Problem \ref{prob:online} can be obtained via a one-dimensional search: calculate $ Q_{d_i}^{*} $ for the $ d_i $-th cache node, and find a cache node with minimum $ Q_{d_i}^{*} $. The computation complexity is low.

\section{Reinforcement Learning Algorithm for Unknown System Statistics}\label{sec:learning}

In Section \ref{sub:linear-app}, the approximate per-file per-region value functions $J_{k,f,c}({B}_{k})$ ($\forall k,f,c$) are evaluated analytically by assuming the knowledge on the distribution of the requesting users $ \mathcal{D} $ and the file popularity $\{p_f|\forall f\}$. In practice, however, the former distribution may not be available at the BS, and the initial estimation of file popularity may not be accurate. In this section, a reinforcement learning algorithm is proposed to estimate the approximate value function for the above practical scenario in an online way.

For elaboration convenience, $\forall c=0,1,2...,N_C$, we define 
$$
 \mu_c \triangleq \min_{\{P_{k},N_{k}\}}  \Pr[\mathbf{l}_{ k } \in \mathcal{C}_{c}] \times\mathbb{E}_{\zeta_k}  \Big[  P_{k}N_{k}+wN_{k} \Big|\mathbf{l}_{k}\in \mathcal{C}_{c} \Big],$$ where the constraints in \eqref{eqn:peak_power} and \eqref{constrain:user} should be satisfied.
Hence, we have $$ \mu_c\approx\Pr[\mathbf{l}_{ k } \in \mathcal{C}_{c}] \times\mathbb{E}_{\zeta_k}  \Big[  F(\theta,P_B)\Big|\mathbf{l}_{k}\in \mathcal{C}_{c} \Big]$$ in the high SINR region. It can be observed from (\ref{eqn:app_J_{k,f,0}}), (\ref{eqn:approx-high}) and (\ref{eqn:low-pop-app}) that, the approximate  per-file per-region value functions  depend on $ \{ \mu_c | \forall c=0,1,2,...,N_C\}$, $\{\upsilon_{k,f,c} |\forall k,f,c\} $ (defined in (\ref{eqn:v_{k,f,c}})) and $\{p_f | \forall f \}$. Instead of the learning of the value functions (e.g., the Q-learning method in \cite{Bertsekas2000Dynamic}) directly, we propose to learn $ \{ \mu_c | \forall c=0,1,2,...,N_C\}$, $\{\upsilon_{k,f,c} |\forall k,f,c\} $ and $\{p_f | \forall f \}$, and calculate the per-file per-region value functions in the following algorithm. To facilitate more efficient learning on files' popularity, we extend the data collection scope from single cell to a network. Specifically, it is assumed that the file popularity is homogeneous within $ N_{cell} $ cells, and each BS can collect the history data of file requests from all these BSs.

\begin{Algorithm}[Reinforcement Learning for Per-File Per-Region Value Functions]\label{alg:learning}
	\ \
	\begin{itemize}
		\item {\bf Step 1}: Let $ t=0 $. Initialize the values of $J_{k,f,c}(B_k), p_f,\mu_c,\upsilon_{k,f,c}$ ($\forall k,f,c$) according to certain assumptions on user arrival and popularity distributions, denoted as $J^0_{k,f,c}(B_k), p_f^0,\mu_c^0,\\\upsilon_{k,f,c}^0$ respectively. 
		
		\item {\bf Step 2}: Let $ t= t+ 1 $ on each new file request arrival.  Update $p_f^t,\mu_c^t,\upsilon_{k,f,c}^t(\forall k,f,c)$ as follows.
		\begin{itemize}
			\item $p_f^t=\frac{t}{t+1}p_f^{t-1}+\frac{1}{t+1}\frac{N_{R,f}^{t}}{\beta N_{cell}T_F^t}$, where $N_{R,f}^{t}$ is the total number of requests on the $ f $-th file in the $N_{cell}$ cells during $T_F^t$ frames.
		
		\item $ \mu_c^t=\frac{t}{t+1}\mu_{c}^{t-1}+\frac{1}{t+1}\mathbf{I}[\mathbf{l}_{m} \in \mathcal{C}_{c}] \times F(\theta_m,P_B) $, where $\mathbf{l}_{m}$ is the location of the requesting user, and $\theta_m =	\mathbb{E}_{\mathbf{h}_{m}} \left[ \log_2 \left(  \frac{||\mathbf{h}_{m}||^2}{N_T \sigma^2_z}\right) \right] $.
		
		\item Moreover,
		\begin{align}
				\upsilon_{k,f,c}^{t}&=\frac{t}{t+1}\upsilon_{k,f,c}^{t-1}+\frac{1}{t+1}\bigg\{\mathbf{I}[\mathbf{l}_{m} \in \mathcal{C}_{f,m}( \mathbf{b}_{f}^c )]\times \widehat{W}_{k+1,f}^{t}( \mathbf{b}_{f}^c )  \nonumber\\ 
				&+\mathbf{I}[\mathbf{l}_{m} \notin \mathcal{C}_{f,m}( \mathbf{b}_{f}^c )]\times \bigg[ \overline{G}_{k,f}^1 \mathbf{I}[R_m \leq R_m^c] +\min \{\overline{G}_{k,f}^2, \overline{G}_{k,f}^3\} \mathbf{I}[R_m > R_m^c] \bigg] \bigg\}, \nonumber
				\end{align}
			
				where $\overline{G}_{k,f}^1,\overline{G}_{k,f}^2,\overline{G}_{k,f}^3$ are given by 
			(\ref{eqn:Q}).

\end{itemize}

	\item {\bf Step 3}: Calculate $J_{k,f,c}^t({B}_{k})$ ($\forall k,f,c$) according to (\ref{eqn:app_J_{k,f,0}}), (\ref{eqn:approx-high}) and (\ref{eqn:low-pop-app}), respectively.
	
	\item {\bf Step 4}: If $\max_{f,c}|J_{k,f,c}^{t}({B}_{k})-J_{k,f,c}^{t-1}({B}_{k})| $ is smaller than one threshold, terminate. Otherwise, go to Step 2.

	\end{itemize}

\end{Algorithm}

	In the above algorithm, the unbiased estimations of $p_f,\mu_c$ and $\upsilon_{k,f,c}(\forall k,f,c)$ are utilized to update the learning results. Hence, the learning procedure always converges, and the mean squared errors of the estimated $p_f^t,\mu_c^t$ and $\upsilon_{k,f,c}^t(\forall k,f,c)$ decrease with the order of $\mathcal{O}(1/t)$.

\section{Simulation Results}\label{sec:sim}

\begin{figure*}
	\centering
	\begin{tabular}{ccc}
		\includegraphics[width=3.4in,height=2.5in]{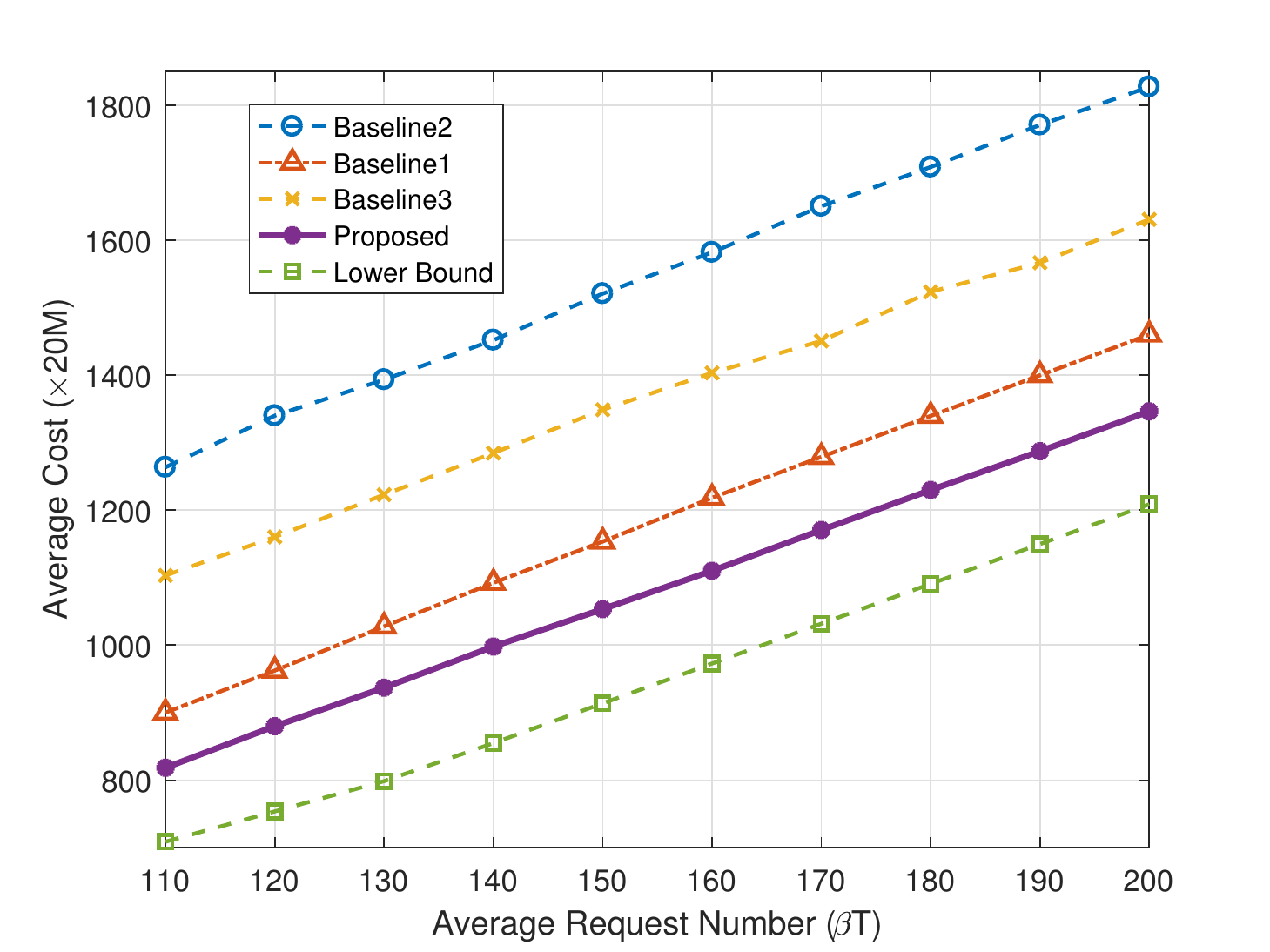}  & 
		\includegraphics[width=3.4in,height=2.5in]{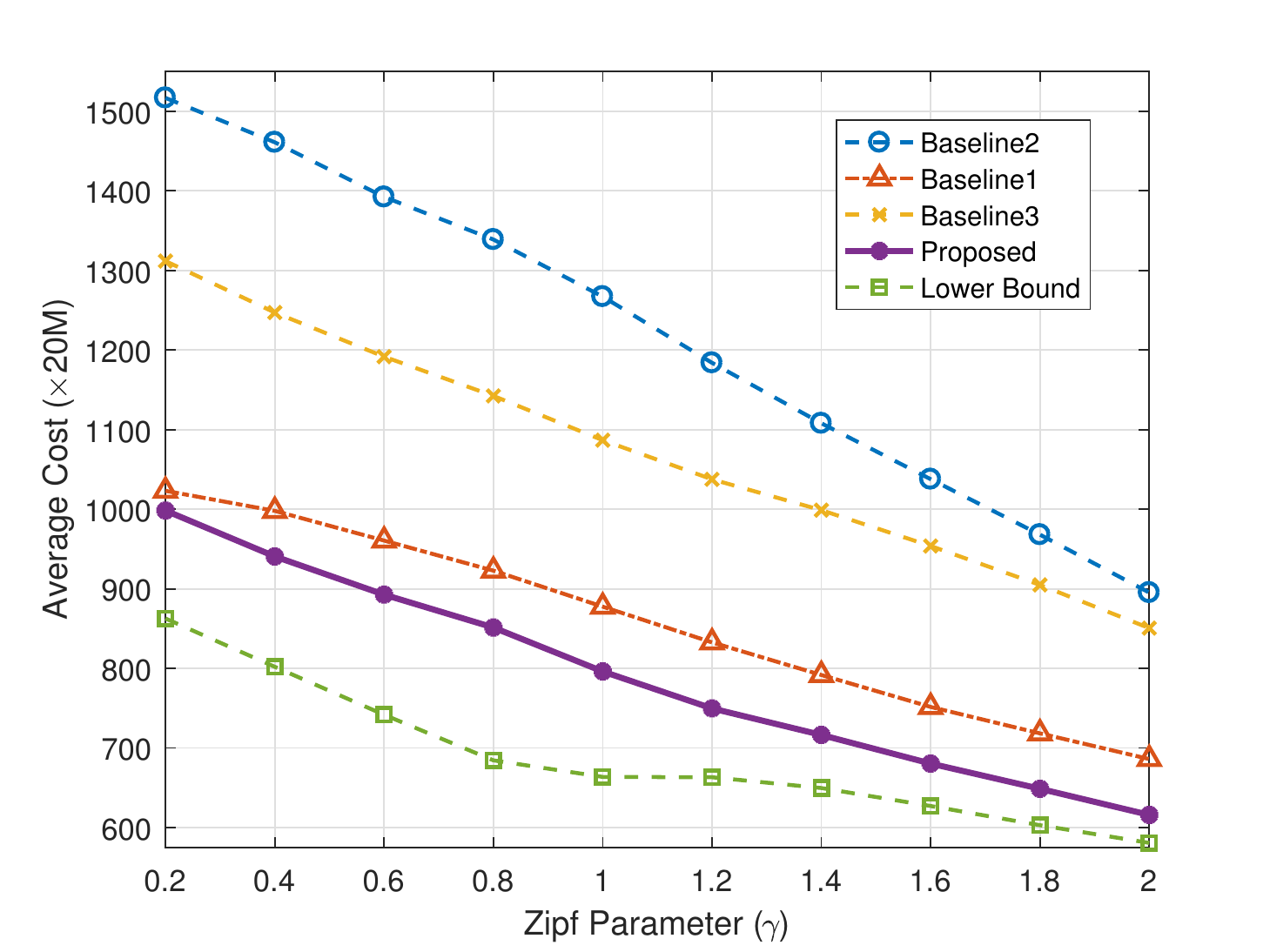}  &  \\ 
		(a) $\gamma=1.2$, $\frac{M_c}{M_F}=0.6 (\forall c)$ & (b) $\frac{M_c}{M_F}=0.6 (\forall c)$, $\beta L=100$&\\	
	\end{tabular}
	\caption{The average total cost versus the expectation of request times and Zipf parameter, where the geometrical distribution of requesting users is uniform, and the file popularity $ \{p_f| \forall f\} $ is known to the BS.}
	\label{Fig:online_sim}
	\vspace{-0.5em}	
\end{figure*}

\begin{figure}[tb]
	\centering
	\begin{tabular}{ccc}
			\includegraphics[width=3.4in,height=2.5in]{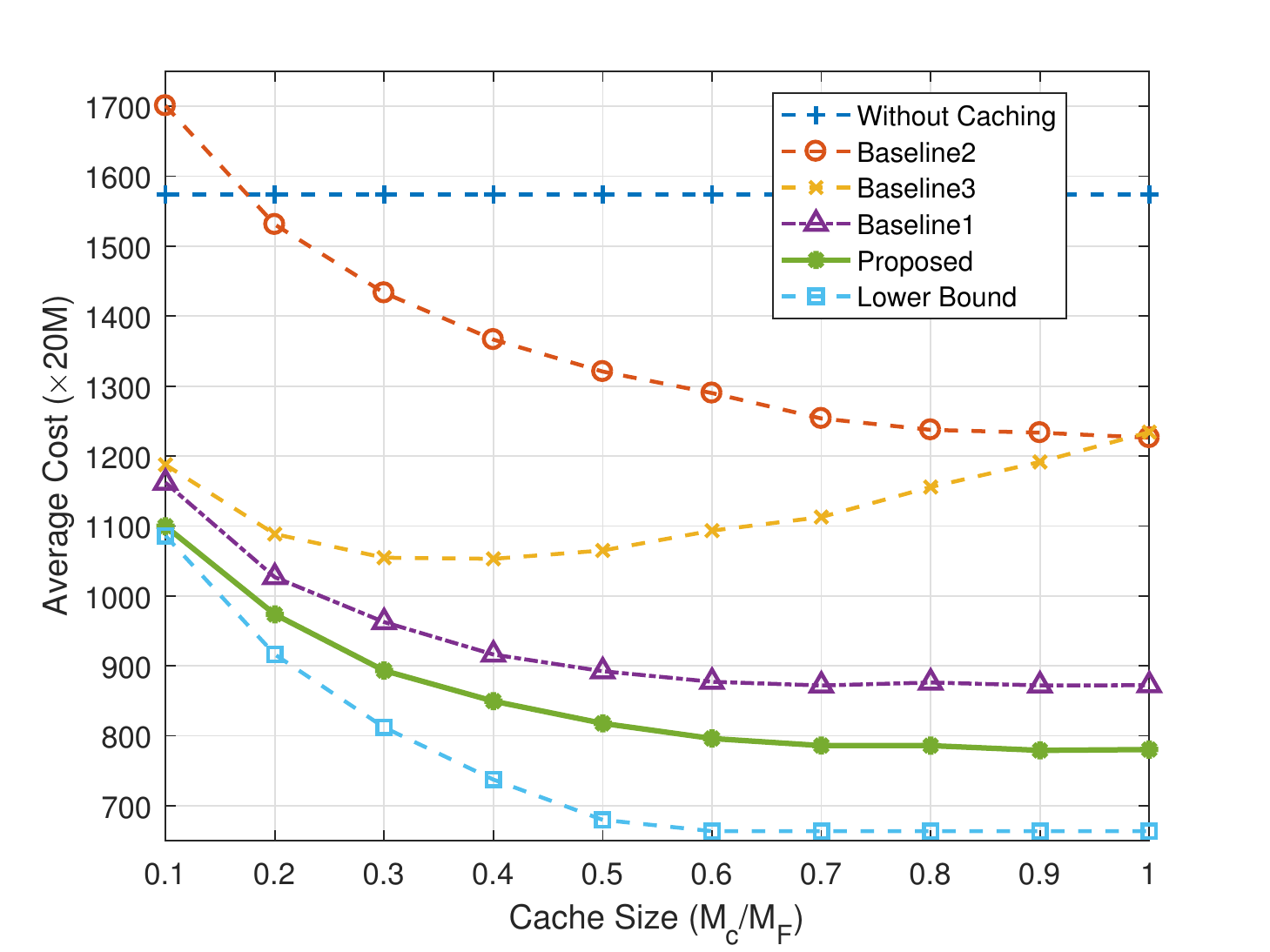} &
			\includegraphics[width=3.4in,height=2.5in]{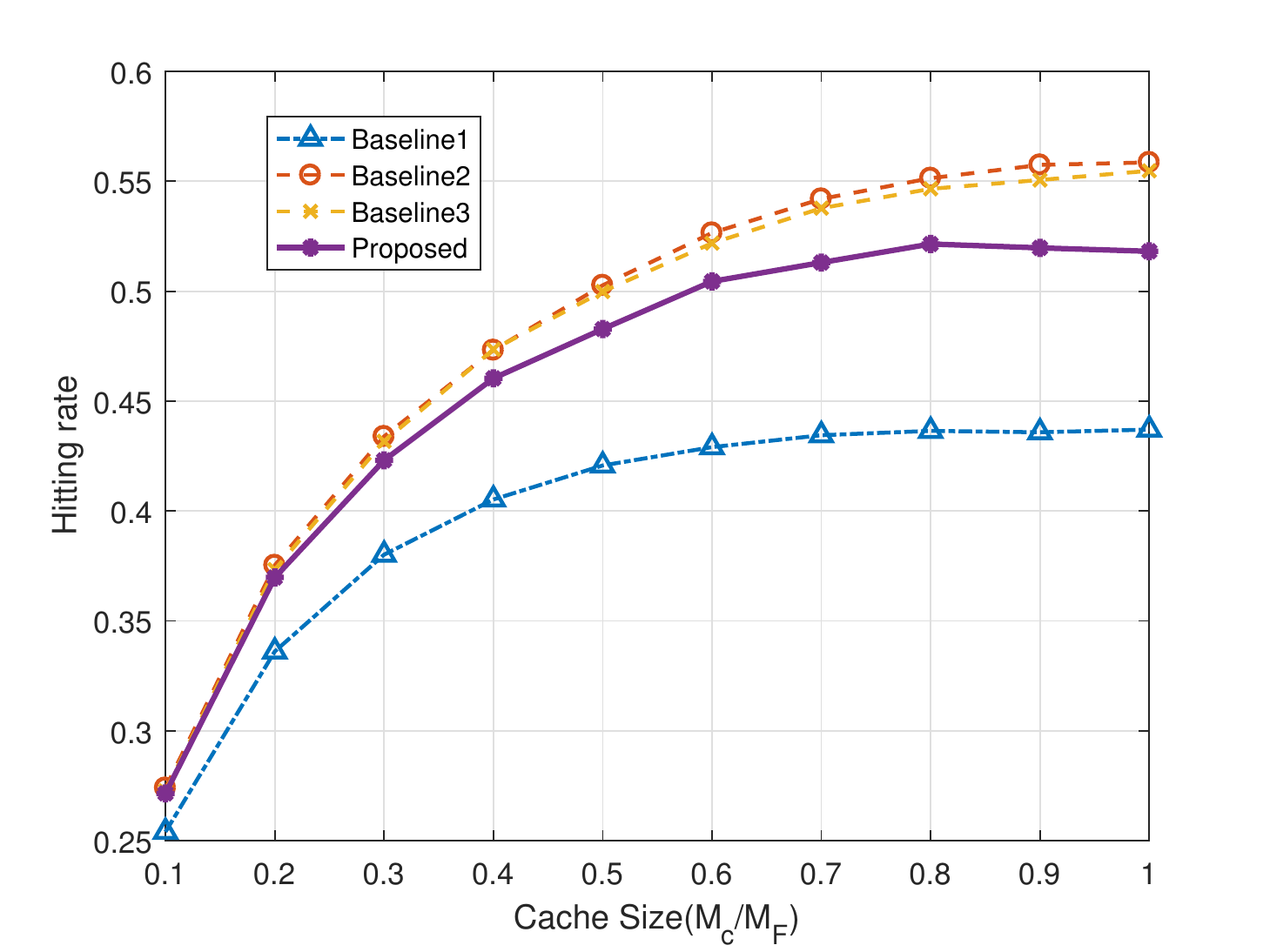}  &  \\ 
				(a) Average total cost versus $\frac{M_c}{M_F}$& (b) Hitting rate versus $\frac{M_c}{M_F}$ &\\
	\end{tabular}

	\caption{The average total cost and hitting rate versus the cache size ($\frac{M_c}{M_F}$), where $\gamma=1$, $\beta L=100$, the distribution of requesting users is uniform, and the file popularity $ \{p_f| \forall f\} $ is known to the BS.}
	\label{fig:cache_size}
\end{figure}


\begin{figure}[tb]
	\centering
	\includegraphics[width=5in,height=2.5in]{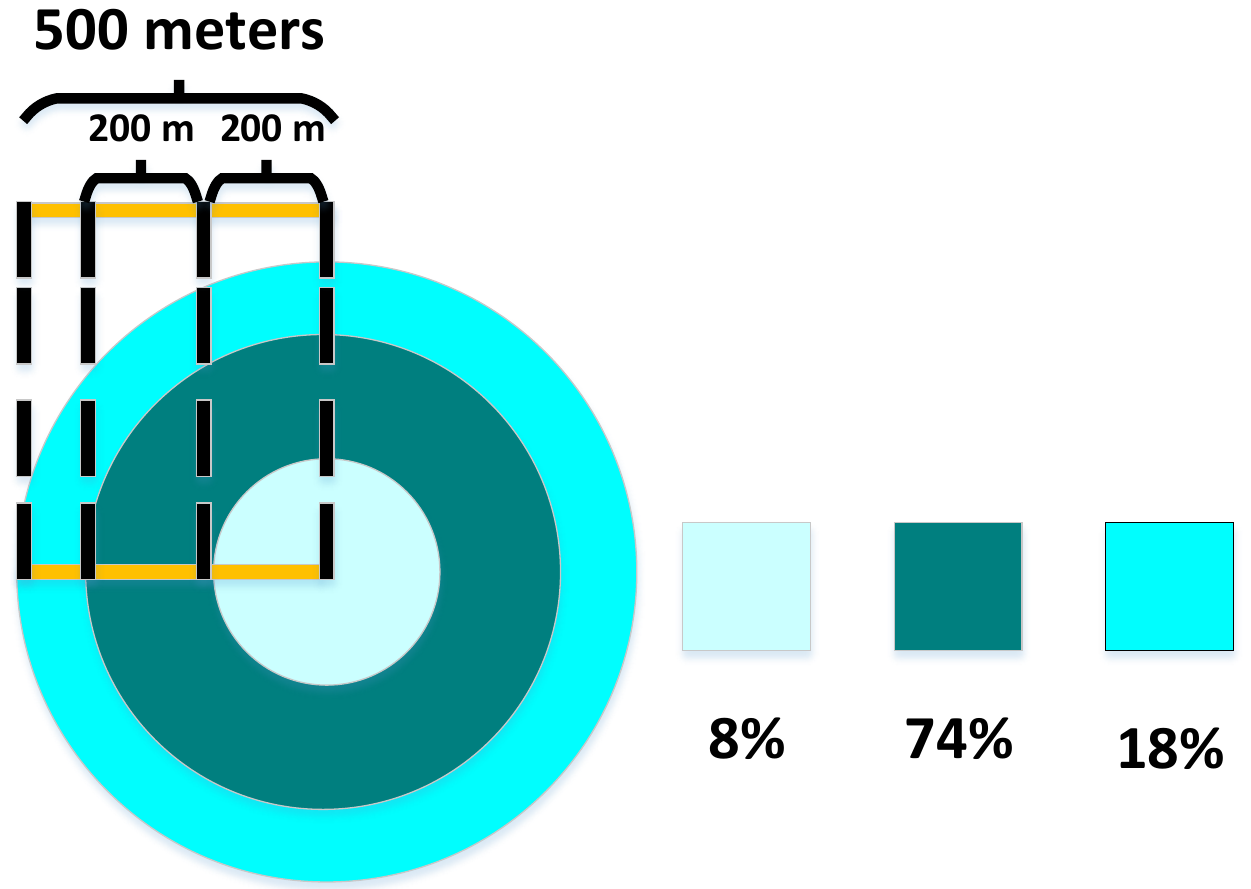}
	\caption{Illustration of a non-uniform geometrical distribution of the requesting users $ \mathcal{D} $, where whole cell area is divided into three regions, the users' distribution in each region is uniform, and the probabilities one requesting user falls into these three regions are 8$\%$, 74$\%$, and 18$\%$ respectively.}
	\label{fig:setting}
\end{figure}

\begin{figure}[tb]
	\centering
	\includegraphics[width=5in,height=2.5in]{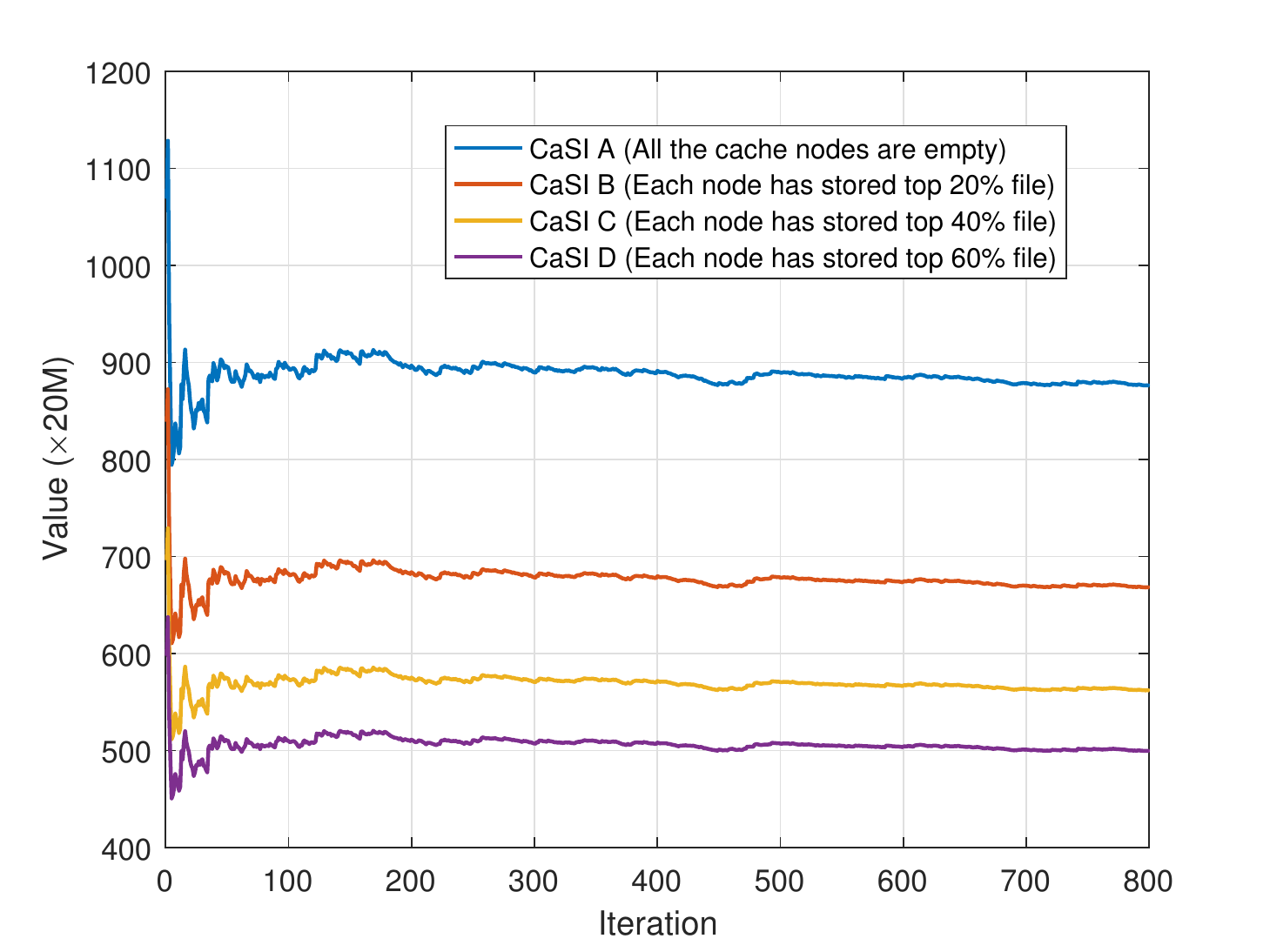}
	\caption{Illustration of the converge of approximated value function via reinforcement learning approach (Algorithm \ref{alg:learning}), when $k=1$, $\gamma=1.3$, $\beta L=100$ and $\frac{M_c}{M_F}=0.6$ ($\forall c$).}
	\label{fig:Converge}
\end{figure}

\begin{figure}[tb]
	\centering
	\includegraphics[width=5in,height=2.5in]{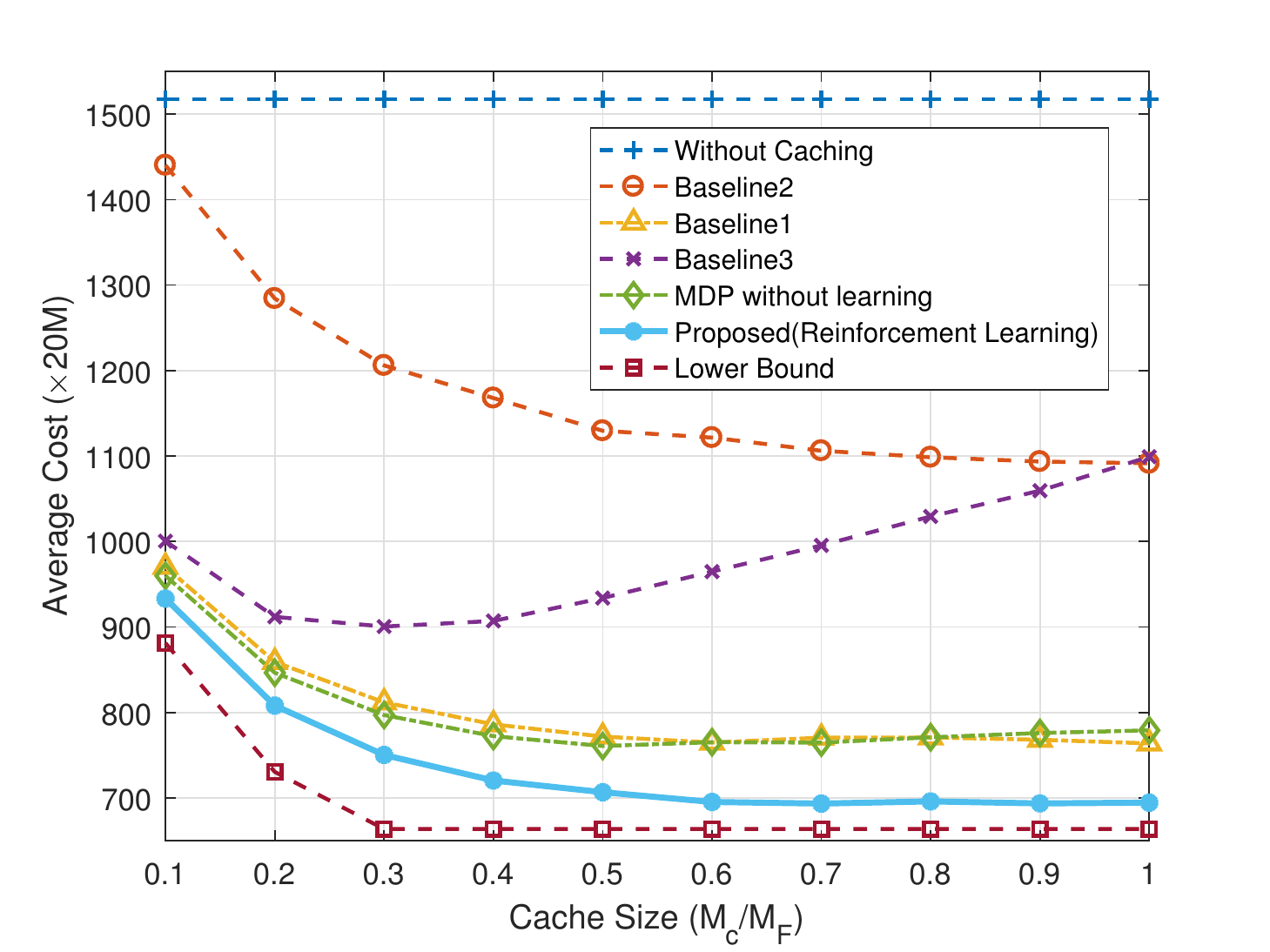}
	\caption{Illustration of the performance of approximated value function via reinforcement learning approach (Algorithm \ref{alg:learning}), when $\gamma=1.3$ and $\beta L=100$.}
	\label{fig:learning}
\end{figure}

In the simulation, the radius of the cell is $ 500 $ meters, $ 21 $ cache nodes are deployed in the cell-edge region, each with a service radius of $ 90 $ meters. The number of antennas at the BS is $8$. The downlink pathloss exponent is $3.76$, and the standard deviation of the shadowing effect is $10dB$. Each file consists of $14M$ information bits, and the transmission bandwidth is $ 20 $MHz. The power constraint
	at the base station is $P_B = 47 dBm$. The lifetime and the frame duration are 24 hours and 10 milliseconds respectively, and the number of frames in the lifetime $ L= 100 \times 60 \times 60 \times 24 \approx 8.6 \times 10^6$. It is assumed that $\{p_f| \forall f\}$ follows a Zipf distribution  with skewness factor $\gamma$ \cite{Web_cache}. The following three baseline schemes are compared with the proposed scheduling scheme.
\begin{Baseline}
The BS only ensures the file delivery to the requesting user in each transmission. The cache nodes with better channel conditions to the BS can also decode the file. The decoded file will be stored at the cache node if there is spare memory or files with lower popularity at the cache node.
\end{Baseline}
\begin{Baseline}The BS ensures that all the cache nodes can decode each file in its first transmission. Files with lowest popularity will be replaced at cache nodes if the caches are fully occupied. 
\end{Baseline}
\begin{Baseline}
		In the first transmission of the $f$-th file, the BS ensures that the cache nodes, where $f$-th file is of high-popularity, can decode the $f$-th file. After that, the BS only ensures the delivery of the $f$-th file to the requesting user. The stored file with the lowest popularity at a cache node will be replaced if the cache node is fully occupied and a newly received file has higher popularity.
\end{Baseline}

With the knowledge on uniform distribution of the requesting users and the file popularity $\{p_f | \forall f\}$ at the BS, the average transmission cost versus the average number of file requests in the whole lifetime $ \beta L $ is illustrated in Fig. \ref{Fig:online_sim} (a), where each cache node can store $ 60 \% $ popular files. It can be observed that the proposed scheme (Lemma \ref{lem:online-scheduling}) consumes less transmission resource than both baselines. The gap between the lower-bound (\ref{eqn:tighter_lower_bound}) and the performance of the proposed algorithm is small, demonstrating the tightness of the lower-bound and the good performance of the proposed algorithm. Moreover, the performance gain tends to be a constant when $ \beta L $ is large. This is because all the four schemes (proposed scheme, Baseline 1, 2 and 3) have the same performance if all the high-popularity files have been stored in cache nodes. In other words, the gain of the proposed scheme lies in the phase of cache placement. 

In Fig. \ref{Fig:online_sim} (b), the performance of the four schemes are compared for different skewness factors of Zipf distribution. The proposed algorithm consumes less transmission resource than both baselines. It can be observed that the gain of the proposed scheme over Baseline 2 is more significant when the popularity of the files is close to uniform. This is because the transmission coupling between the high-popularity and low-popularity files are better exploited in the proposed scheme. Moreover, the gap between  the cost lower-bound and the cost of the proposed algorithm is smaller for larger skewness factor. In other words, the proposed scheme is close to the optimal solution when the Zipf distribution is steeper. 

In Fig. \ref{fig:cache_size} (a), the effect of the cache size is evaluated for the skewness factor $ \gamma =1 $, where the average system cost versus the $ M_c / M_F $ (the ratio of high-popularity files) is plotted. It can be observed that when increasing the cache size, more traffic can be offloaded to cache nodes, which leads to less transmission resource consumption at the BS in all three schemes. Moreover, the reduction of average system cost becomes slow in all the schemes and the cost lower-bound when $ M_c / M_F > 50 \% $. This is because the increased cache space is for the files with small request probability, and the average transmission resource saving is also small. It can also be observed that the gap between the lower-bound and the proposed scheme is small, and the proposed scheme has better utilization of the caches than both baseline schemes especially for large cache size. Finally, the performance of Baseline 3 is close to that of Baseline 1 with scarce cache space, and close to that of Baseline 2 with rich cache space.

In Fig. \ref{fig:cache_size} (b), the hitting rate,  defined as $$\text{Hitting rate}\triangleq\mathbb{E}_{\{S_{\tau}|\forall \tau\},N_R}^{\Omega}\bigg[\frac{1}{N_R}\sum_{\tau=1}^{N_R}\mathbf{I}  [\mathbf{l}_{\tau} \in \mathcal{C}_{\mathcal{A}_{\tau},\tau}]\bigg],$$
are compared among the four schemes.
It can be observed that  the BS becomes more conservative on exploiting the cache nodes in traffic offloading from Baseline 2, 3, proposed scheme to Baseline 1. Their hitting rates are therefore in descending order. Although Baseline 2 and 3 have higher hitting rates than that of the proposed scheme, their costs are greater than the proposed scheme, as illustrated in Fig. \ref{fig:cache_size} (a). Hence, maximizing the hitting rate is not a good strategy if the concern is average transmission cost of the BS.


In the simulation of the previous figures, the distribution of requesting users is uniform, and the file popularity is also known to the BS. When the non-uniform spatial distribution of the requesting users $ \mathcal{D}$ and the file popularity $\{p_f | \forall f\}$ are not available at the BS, the reinforcement learning algorithm (Algorithm \ref{alg:learning}) can be used in the downlink scheduling. In simulation of both Fig. \ref{fig:Converge} and \ref{fig:learning}, the actual distribution of requesting users is non-uniform as in Fig. \ref{fig:setting}. Fig. \ref{fig:Converge} illustrates the convergence of reinforcement learning algorithm (Algorithm \ref{alg:learning}) for some CaSIs. In Fig. \ref{fig:learning},  the performance of three baselines, the proposed scheme in Lemma \ref{lem:online-scheduling} with inaccurate values of $\{p_f| \forall f\}$ and the wrong assumption of uniform user distribution, and the proposed scheme with reinforcement learning (Algorithm \ref{alg:learning}) are compared. It can be observed that the proposed learning algorithm has the best performance. The gap between the cost lower-bound and the proposed learning algorithm is small. Thus, the proposed learning algorithm is close to the optimal solution.

 Finally, in order to compare the computation times of value function evaluation between the proposed scheme and the optimal solution, we consider the following simplified scenario\footnote{The computation of the optimal solution is prohibitive with the previous simulation configuration.}. There are 2 cache nodes ($N_C=2$) in the cell, each can store $1$ popular file. The library is with $M_F=2$ popular files. The computation time of the optimal value iteration is around $10.7$ times larger than that of the proposed schema. This ratio will be much larger if we increase $M_F$ or $N_C$.

\section{Conclusion}\label{sec:con}
In this paper, we consider the downlink file transmission with the assistance of cache nodes within a finite lifetime. The BS multicasts files to the requesting users and the selected cache nodes reactively, and the cache nodes with decoded files can help offload the traffic via other air interfaces. We formulate the joint optimization of such file placement and delivery as a dynamic programming problem with a random number of stages, and propose an asymptotically optimal way to transform the original problem into a finite-horizon MDP with a fixed number of stages. In order to avoid the curse of dimensionality, we also introduce a low-complexity sub-optimal solution based on linear approximation of the value functions, which can be calculated analytically. The bound of the approximation error is derived. Finally, a reinforcement learning algorithm is proposed to obtain the approximate value functions in the practical scenario, where system statistics is not available.

\section*{Appendix A: Proof of Lemma 1}
 The expression of $W_k(S_k)$ can be rewritten as
 \begin{align} 
W_k(S_k)&=\min\limits_{\{\Omega_k, \Omega_{k+1}, ...\}}\mathbb{E}_{\mathcal{A},\zeta,N_R}\{\sum_{\tau=k}^{N_R}g_{\tau}(S_{\tau}, \Omega_{\tau})\mathbf{I}(k \leq N_R)| S_k\}\nonumber \\
&=\min\limits_{\{\Omega_k, \Omega_{k+1}, ...\}}\mathbb{E}_{\mathcal{A},\zeta,N_R}\{\sum_{\tau=k}^{N_R}g_{\tau}(S_{\tau}, \Omega_{\tau}) |N_R\geq k , S_k\}\Pr(N_R\geq k)\nonumber. 
\end{align}
Given system state $S_k$, we have
 \begin{align} 
W_k(S_k)
&=\min\limits_{\{\Omega_k, \Omega_{k+1}, ...\}}\!\!\mathbb{E}_{\mathcal{A},\zeta,N_R}\{g_k(S_{k}, \Omega_{k})+\!\!\!\!\sum_{\tau=k+1}^{N_R}\! g_{\tau}(S_{\tau}, \Omega_{\tau})\mathbf{I}(k+1 \leq N_R)|N_R\geq k\}\!\Pr(N_R\geq k)
\nonumber \\
&=\min\limits_{\{\Omega_k, \Omega_{k+1}, ...\}}\bigg\{ g_k(S_k,\Omega_k)\Pr(N_R\geq k)\nonumber \\
&+\sum\limits_{S_{k+1}}\Pr[S_{k+1}|S_k,\Omega_{k}]  \mathbb{E}_{\mathcal{A},\zeta,N_R}\{\sum_{\tau=k+1}^{N_R}g_{\tau}(S_{\tau}, \Omega_{\tau})\mathbf{I}(k+1 \leq N_R)|N_R\geq k\}\Pr(N_R\geq k) \! \bigg\} \nonumber\\
&= \min\limits_{\{\Omega_k, \Omega_{k+1}, ...\}}\bigg\{ g_k(S_k,\Omega_k)\Pr(N_R\geq k)  +\sum\limits_{S_{k+1}}\Pr[S_{k+1}|S_k,\Omega_{k}] W_{k+1}(S_{k+1}) \bigg\},
\end{align}

where the last step is because
\begin{align}
W_{k+1}(S_{k+1})&=\min\limits_{\{\Omega_{k+1}, ...\}}\!\bigg\{\mathbb{E}_{\mathcal{A},\zeta,N_R}\{\sum_{\tau=k+1}^{N_R}g_{\tau}(S_{\tau}, \Omega_{\tau})\mathbf{I}(k+1 \leq N_R)|N_R\geq k, S_{k+1}\}\Pr(N_R\geq k)\nonumber \\
&+\underbrace{\mathbb{E}_{\mathcal{A},\zeta,N_R, S_{k+1}}\{\sum_{\tau=k+1}^{N_R}g_{\tau}(S_{\tau}, \Omega_{\tau})\mathbf{I}(k+1 \leq N_R)                                                                                                                     |N_R< k,S_{k+1}\}\Pr(N_R< k)}_{=0}\bigg\}. \nonumber
\end{align}

\section*{Appendix B: Proof of Lemma \ref{lem:reduce stage number}}
\subsubsection{Proof of upper-bound}
Because  $\overline{g}_{max} =\max\limits_{\rho_k,{B}_{k}}\min\limits_{\Omega_k} \mathbb{E}_{\mathcal{A},\eta_k}\Big\{\! g_{k} ({S}_{k}, \Omega_{k})\! \Big \} \geq \mathbb{E}_{\mathcal{A},\zeta_k}\Big \{\! g_{k} ({S}_{k}, \Omega^{*}_{k})\! \Big\}$, $\forall {S}_{k}, k$, we have
	\begin{align*}
	\widetilde{W}_{k}({B}_{k}) &= \mathbb{E}_{\mathcal{A},\zeta}\bigg \{ \sum_{\tau=k}^{L}  
	g_{\tau} ({S}_{\tau}, \Omega^{*}_{\tau}) \Pr(\tau \leq N_R)\bigg| {B}_{k} \bigg\}\\
	&\leq \min\limits_{\{\Omega_{\tau}|\forall \tau\}} \mathbb{E}_{\mathcal{A},\zeta}\bigg \{ \sum_{\tau=k}^{M_R^{\epsilon}}  
	g_{\tau} ({S}_{\tau}, \Omega_{\tau}) \Pr(\tau \leq N_R) \bigg| {B}_{k} \bigg\}+ \sum_{\tau=M_R^{\epsilon}+1}^{L}  
	\overline{g}_{max} \Pr (\tau \leq N_R), \nonumber
	\end{align*}

where the upper-bound is due to non-optimal policy. Similar to \cite{Lv2019,WANG2017}, the analytical expression of $\overline{g}_{max}$ is straightforward.

\subsubsection{Proof of lower-bound} 

	\begin{align*}
	\widetilde{W}_{k}({B}_{k}) 	\!\geq\! \mathbb{E}_{\mathcal{A},\zeta}\bigg \{\! \sum_{\tau=k}^{M_R^{\epsilon}} \!  
	g_{\tau} ({S}_{\tau}, \Omega^{*}_{\tau}) \Pr(\tau \leq N_R)\bigg| {B}_{k} \!\bigg\} \!\geq\! \min\limits_{\{\Omega_{\tau}|\forall \tau\}} \!\!\mathbb{E}_{\mathcal{A},\zeta}\bigg \{\! \sum_{\tau=k}^{M_R^{\epsilon}} \! 
	g_{\tau} ({S}_{\tau}, \Omega_{\tau})\! \Pr(\tau \leq N_R)\bigg| {B}_{k} \!\bigg\}. \nonumber
	\end{align*}
	
\section*{Appendix C: Proof of Lemma \ref{lem:upper-bound}}

Due to page limitation, we only provide the sketch of the proof. In order to prove Lemma \ref{lem:upper-bound}, we only need to prove $ \mathcal{L}_k ({B}_k) \leq \sum_{f,c} J_{k,f,c} ({B}_k) $ for all $ k $ and $ {B}_k $. Hence, we first introduce the following heuristic scheduling policy $ \{\Omega^{\star}_{\tau} | \tau=1,2,...,M_R^{\epsilon} \} $.
\begin{Policy}[Heuristic Scheduling Policy] \label{policy:heuristic}
	The scheduling actions from the first request to the $ M_R^{\epsilon} $-th request are provided below. 
	\begin{itemize}
	\item When (a) the requesting user is in the coverage area of one cache node and (b) the requested file is low-popular for that cache node and it has not been cached, the BS only deliver the file to the requesting user. Thus in this case, $ \Delta \mathcal{B}_{f,k}^{c} = 0$ and
	$
	\{P_{k}^{\star}, N_{k}^{\star} \} = \arg \min (P_{k}N_{k} + w N_{k}) \mathbf{I} [\mathbf{l}_{k} \notin \mathcal{C}_{\mathcal{A}_{k},k}] $ , where the constraints in \eqref{eqn:peak_power} and  \eqref{constrain:user} should be satisfied.
	
	\item Otherwise, the transmission is optimized such that the average total cost on remaining transmissions is minimized. Thus, let $ \Omega_{\tau}^{\star,H} $ be the scheduling policy on remaining transmissions, we have
	\begin{align}
 \{\Omega_{\tau}^{\star,H} | \forall \tau =1,...,M_R^{\epsilon}\} 
 =  \arg \min \mathbb{E}_{\mathcal{A},\zeta} \sum_{\tau=1}^{M_R^{\epsilon}} \sum_{c=0}^{N_C}
g_{\tau} ({S}_{\tau}, \Omega_{\tau}) \Pr(\tau \leq N_R) \mathbf{I}(\mathcal{A}_{\tau} \in \mathcal{F}^H_{c} \cap \mathbf{l}_{\tau} \in \mathcal{C}_c) \nonumber
\end{align}

	where the constraints in (\ref{eqn:peak_power},\ref{constrain:user},\ref{constrain:cache},\ref{constrain:cache_action1},\ref{constrain:cache_action2}) should be satisfied, and $ \mathcal{F}_0^H\triangleq\{1,...,M_F\} $ for notation convenience.
	\end{itemize}
\end{Policy}
With the Policy \ref{policy:heuristic}, let $ \mathcal{U}^{\star}_{k}({B}_{k}) $ be the average cost from the $ k $-th stage to the $ M_R^{\epsilon} $-th stage. Note that the optimal policy is used in $ \mathcal{L}_k $, $ \mathcal{L}_k ({B}_{k}) \leq \mathcal{U}^{\star}_{k} ({B}_{k}) $. Moreover, it is clear that $ \mathcal{U}^{\star}_{k} ({B}_{k}) \!\leq \sum_{f,c} J_{k,f,c} ({B}_k) $, the conclusion of Lemma \ref{lem:upper-bound} can be obtained.

\section*{Appendix D: Proof of Lemma \ref{lem:first lower-bound of value function}}
Due to page limitation, we only provide the sketch of the proof. First,  $\widetilde{W}_{k} ({B}_{k}) \geq \mathcal{L}_k ({B}_{k})$ is given by (\ref{eqn:lower-bound:reduce_system_stage}). In order to prove $\widetilde{\mathcal{L}}_k^1({B}_{k})\leq  \mathcal{L}_k({B}_{k})$, the approach of mathematical induction shall be used. 
\begin{itemize}
	\item {\bf{Step 1}}: When $\tau=M_R^{\epsilon}$, 	
	\begin{align*}
\widetilde{\mathcal{L}}_{M_R^{\epsilon}}^1({B}_{M_R^{\epsilon}})\!&=\!\bar{g}_{min}({B}_{M_R^{\epsilon}})\!\Pr(M_R^{\epsilon} \leq N_R) \!=\!\min_{\Omega_{M_R^{\epsilon}}}\!\mathbb{E}_{\mathcal{A},\zeta}\{ g_{M_R^{\epsilon}}({S}_{M_R^{\epsilon}},\Omega_{M_R^{\epsilon}})|B_{M_R^{\epsilon}} \}\!\Pr(M_R^{\epsilon} \leq N_R)\\&=\!\mathbb{E}_{\mathcal{A},\zeta,N_R}\bigg \{  
g_{M_R^{\epsilon}} ({S}_{M_R^{\epsilon}}, \Omega^{\dagger}_{M_R^{\epsilon}}) \mathbf{I}(M_R^{\epsilon} \leq N_R)|B_{M_R^{\epsilon}}\bigg\} =\mathcal{L}_{M_R^{\epsilon}}({B}_{M_R^{\epsilon}}).
	\end{align*}

	\item {\bf{Step 2}}: Suppose the lower-bound holds for $\tau=k$. When $\tau=k-1$,
		\begin{eqnarray}
		&\widetilde{\mathcal{L}}_{k-1}^1\!(\!{B}_{k-1}\!)\!\!=\bar{g}_{min}({B}_{k-1})\Pr(k-1\leq N_R)+\widetilde{\mathcal{L}}_{k}^1(\pi({B}_{k-1})) \nonumber \\
		&{\mathcal{L}}_{k-1}\!(\!{B}_{k-1}\!)\!\!=\mathbb{E}_{\mathcal{A},\zeta,N_R}\!\bigg \{\!  
			g_{k-1} \!(S_{k-1}, \Omega^{\dagger}_{k-1}) \mathbf{I}(k-1 \leq N_R)|B_{k-1}\!\!\bigg\}
		\!\!	+ \!\!\sum_{{S}_{k}}\!{\mathcal{L}}_{k}(\!{B}_{k}\!)\!\Pr\! \bigg(\!\!{B}_{k}|{B}_{k-1},\Omega^{\dagger}_{k-1}\!\!\bigg)\nonumber\\
		&\bar{g}_{min}({B}_{k-1})\Pr(k-1 \leq N_R) \leq \mathbb{E}_{\mathcal{A},\zeta,N_R}\bigg \{  
		g_{k-1} ({S}_{k-1}, \Omega^{\dagger}_{k-1}) \mathbf{I}(k-1 \leq N_R)|B_{k-1}\bigg\}.\nonumber
		\end{eqnarray}
	
Given ${B}_{k-1}$, we have $\widetilde{\mathcal{L}}_{k}^1(\pi({B}_{k-1}))\leq  \mathcal{L}_{k}(\pi({B}_{k-1}))\leq  \mathcal{L}_{k}({B}_{k}), \forall {B}_{k}$. Hence, 
$
\widetilde{\mathcal{L}}_{k}^1({B}_{k}^{\pi}) \leq \sum_{{B}_{k}}{\mathcal{L}}_{k}({B}_{k})\Pr \Big({B}_{k}|{B}_{k-1},\Omega^{\dagger}_{k-1}\Big).
$
As a result, $\widetilde{\mathcal{L}}_{k-1}^1({B}_{k-1})\leq  \mathcal{L}_{k-1}({B}_{k-1})$, $\forall {B}_{k-1}$. 
\end{itemize}

\bibliographystyle{IEEEtran}
\bibliography{Cache2}

\begin{thebibliography}{10}
\providecommand{\url}[1]{#1}
\csname url@samestyle\endcsname
\providecommand{\newblock}{\relax}
\providecommand{\bibinfo}[2]{#2}
\providecommand{\BIBentrySTDinterwordspacing}{\spaceskip=0pt\relax}
\providecommand{\BIBentryALTinterwordstretchfactor}{4}
\providecommand{\BIBentryALTinterwordspacing}{\spaceskip=\fontdimen2\font plus
\BIBentryALTinterwordstretchfactor\fontdimen3\font minus
  \fontdimen4\font\relax}
\providecommand{\BIBforeignlanguage}[2]{{%
\expandafter\ifx\csname l@#1\endcsname\relax
\typeout{** WARNING: IEEEtran.bst: No hyphenation pattern has been}%
\typeout{** loaded for the language `#1'. Using the pattern for}%
\typeout{** the default language instead.}%
\else
\language=\csname l@#1\endcsname
\fi
#2}}
\providecommand{\BIBdecl}{\relax}
\BIBdecl

\bibitem{Ruiwang2018_GLOBECOM}
B.~Lv, R.~Wang, Y.~Cui, and H.~Tan, ``Joint optimization of file placement and
  delivery in cache-assisted wireless networks,'' in \emph{2018 IEEE Global
  Commun. Conf. (GLOBECOM)}, Dec. 2018, pp. 1--7.

\bibitem{Leconte2016}
M.~Leconte, G.~Paschos, L.~Gkatzikis, M.~Draief, S.~Vassilaras, and
  S.~Chouvardas, ``Placing dynamic content in caches with small population,''
  in \emph{2016 IEEE Intl. Conf. on Computer Commun.(INFOCOM)}, Apr. 2016, pp.
  1--9.

\bibitem{W.Choi2016}
S.~H. Chae and W.~Choi, ``Caching placement in stochastic wireless caching
  helper networks: Channel selection diversity via caching,'' \emph{IEEE Trans.
  Wireless Commun.}, vol.~15, no.~10, pp. 6626--6637, Oct. 2016.

\bibitem{K.B.Huang2017}
J.~Wen, K.~Huang, S.~Yang, and V.~O.~K. Li, ``Cache-enabled heterogeneous
  cellular networks: Optimal tier-level content placement,'' \emph{IEEE Trans.
  Wireless Commun.}, vol.~16, no.~9, pp. 5939--5952, Sept. 2017.

\bibitem{JunZhang2016}
R.~Wang, X.~Peng, J.~Zhang, and K.~B. Letaief, ``Mobility-aware caching for
  content-centric wireless networks: modeling and methodology,'' \emph{IEEE
  Commun. Mag.}, vol.~54, no.~8, pp. 77--83, Aug. 2016.

\bibitem{Shin2019}
H.~{Song}, S.~H. {Chae}, W.~{Shin}, and S.~{Jeon}, ``Predictive caching via
  learning temporal distribution of content requests,'' \emph{IEEE Commun.
  Lett.}, pp. 1--1, 2019.

\bibitem{cui2019-TCOM}
C.~{Ye}, Y.~{Cui}, Y.~{Yang}, and R.~{Wang}, ``Optimal caching designs for
  perfect, imperfect, and unknown file popularity distributions in large-scale
  multi-tier wireless networks,'' \emph{IEEE Trans. Commun.}, vol.~67, no.~9,
  pp. 6612--6625, Sep. 2019.

\bibitem{tao2019}
F.~{Song}, J.~{Li}, M.~{Ding}, L.~{Shi}, F.~{Shu}, M.~{Tao}, W.~{Chen}, and
  H.~V. {Poor}, ``Probabilistic caching for small-cell networks with
  terrestrial and aerial users,'' \emph{IEEE Trans. Veh. Technol.}, vol.~68,
  no.~9, pp. 9162--9177, Sep. 2019.

\bibitem{coded_cache_1}
M.~A. Maddah-Ali and U.~Niesen, ``Fundamental limits of caching,'' \emph{IEEE
  Trans. Inf. Theory}, vol.~60, no.~5, pp. 2856--2867, May 2014.

\bibitem{M.Tao2017}
X.~Xu and M.~Tao, ``Modeling, analysis, and optimization of coded caching in
  small-cell networks,'' \emph{IEEE Trans. Commun.}, vol.~65, no.~8, pp.
  3415--3428, Aug. 2017.

\bibitem{Tao2016}
M.~Tao, E.~Chen, H.~Zhou, and W.~Yu, ``Content-centric sparse multicast
  beamforming for cache-enabled cloud ran,'' \emph{IEEE Trans. Wireless
  Commun.}, vol.~15, no.~9, pp. 6118--6131, Sept. 2016.

\bibitem{cui2016}
B.~Zhou, Y.~Cui, and M.~Tao, ``Stochastic content-centric multicast scheduling
  for cache-enabled heterogeneous cellular networks,'' \emph{IEEE Trans.
  Wireless Commun.}, vol.~15, no.~9, pp. 6284--6297, Sept. 2016.

\bibitem{Ansari2016}
X.~Huang and N.~Ansari, ``Content caching and user scheduling in heterogeneous
  wireless networks,'' in \emph{2016 IEEE Global Commun. Conf. (GLOBECOM)},
  Dec. 2016, pp. 1--6.

\bibitem{cui2016gc}
Y.~Cui, F.~Lai, S.~Hanly, and P.~Whiting, ``Optimal caching and user
  association in cache-enabled heterogeneous wireless networks,'' in \emph{2016
  IEEE Global Commun. Conf. (GLOBECOM)}, Dec. 2016, pp. 1--6.

\bibitem{Cui2017}
Y.~Cui and D.~Jiang, ``Analysis and optimization of caching and multicasting in
  large-scale cache-enabled heterogeneous wireless networks,'' \emph{IEEE
  Trans. Wireless Commun.}, vol.~16, no.~1, pp. 250--264, Jan. 2017.

\bibitem{Lv2019}
B.~{Lv}, L.~{Huang}, and R.~{Wang}, ``Joint downlink scheduling for file
  placement and delivery in cache-assisted wireless networks with finite file
  lifetime,'' \emph{IEEE Trans. Commun.}, vol.~67, no.~6, pp. 4177--4192, Jun.
  2019.

\bibitem{Moghadari2013}
M.~Moghadari, E.~Hossain, and L.~B. Le, ``Delay-optimal distributed scheduling
  in multi-user multi-relay cellular wireless networks,'' \emph{IEEE Trans.
  Commun.}, vol.~61, no.~4, pp. 1349--1360, Apr. 2013.

\bibitem{cui2010}
Y.~Cui and V.~K.~N. Lau, ``Distributive stochastic learning for delay-optimal
  {OFDMA} power and subband allocation,'' \emph{IEEE Trans. Signal Process.},
  vol.~58, no.~9, pp. 4848--4858, Sept. 2010.

\bibitem{Dechene2010}
D.~J. {Dechene} and A.~{Shami}, ``Energy efficient quality of service traffic
  scheduler for {MIMO} downlink {SVD} channels,'' \emph{IEEE Trans. Wireless
  Commun.}, vol.~9, no.~12, pp. 3750--3761, Dec. 2010.

\bibitem{Wang2013}
R.~Wang and V.~K.~N. Lau, ``Delay-aware two-hop cooperative relay
  communications via approximate {MDP} and stochastic learning,'' \emph{IEEE
  Trans. Inf. Theory}, vol.~59, no.~11, pp. 7645--7670, Nov. 2013.

\bibitem{RuiWang2011}
R.~Wang, V.~K.~N. Lau, and Y.~Cui, ``Queue-aware distributive resource control
  for delay-sensitive two-hop {MIMO} cooperative systems,'' \emph{IEEE Trans.
  Signal Process.}, vol.~59, no.~1, pp. 341--350, Jan. 2011.

\bibitem{cui2012-TIT}
Y.~{Cui}, V.~K.~N. {Lau}, R.~{Wang}, H.~{Huang}, and S.~{Zhang}, ``A survey on
  delay-aware resource control for wireless systems—large deviation theory,
  stochastic lyapunov drift, and distributed stochastic learning,'' \emph{IEEE
  Trans. Inf. Theory}, vol.~58, no.~3, pp. 1677--1701, Mar. 2012.

\bibitem{Han2018}
Z.~{Han}, H.~{Tan}, R.~{Wang}, S.~{Tang}, and F.~C.~M. {Lau}, ``Online learning
  based uplink scheduling in hetnets with limited backhaul capacity,'' in
  \emph{2018 IEEE Intl. Conf. on Computer Commun. (INFOCOM)}, April 2018, pp.
  2348--2356.

\bibitem{huang2019mdpbased}
S.~Huang, B.~Lv, and R.~Wang, ``{MDP}-based scheduling design for mobile-edge
  computing systems with random user arrival,'' \emph{arXiv preprint
  arXiv:1904.13024v2}, 2019.

\bibitem{Han2016}
Z.~{Han}, H.~{Tan}, G.~{Chen}, R.~{Wang}, Y.~{Chen}, and F.~C.~M. {Lau},
  ``Dynamic virtual machine management via approximate markov decision
  process,'' in \emph{2016 IEEE Intl. Conf. on Computer Commun. (INFOCOM)},
  Apr. 2016, pp. 1--9.

\bibitem{Han2019}
Z.~{Han}, H.~{Tan}, R.~{Wang}, G.~{Chen}, Y.~{Li}, and F.~C.~M. {Lau},
  ``Energy-efficient dynamic virtual machine management in data centers,''
  \emph{IEEE/ACM Trans. Netw.}, vol.~27, no.~1, pp. 344--360, Feb. 2019.

\bibitem{lv2019cooperative}
B.~Lv, Y.~Hong, H.~Tan, Z.~Han, and R.~Wang, ``Cooperative job dispatching in
  edge computing network with unpredictable uploading delay,'' \emph{arXiv
  preprint arXiv:1912.10732}, 2019.

\bibitem{Shewhart2011Approximate}
W.~B. Powell, \emph{Approximate Dynamic Programming: Solving the Curses of
  Dimensionality, Second Edition}.\hskip 1em plus 0.5em minus 0.4em\relax
  Hoboken, NJ, USA: Wiley, 2011.

\bibitem{FHMDP}
Q.~Jia, ``A potential-based method for finite-stage markov decision process,''
  in \emph{2008 American Control Conf. (ACC)}, Jun. 2008, pp. 5029--5034.

\bibitem{May2016}
K.~Poularakis, G.~Iosifidis, I.~Pefkianakis, L.~Tassiulas, and M.~May, ``Mobile
  data offloading through caching in residential 802.11 wireless networks,''
  \emph{IEEE Trans. Netw. Service Manag.}, vol.~13, no.~1, pp. 71--84, Mar.
  2016.

\bibitem{Molisch2016}
M.~Ji, G.~Caire, and A.~F. Molisch, ``Wireless device-to-device caching
  networks: Basic principles and system performance,'' \emph{IEEE J. Sel. Areas
  Commun.}, vol.~34, no.~1, pp. 176--189, Jan. 2016.

\bibitem{Web_cache}
L.~Breslau, P.~Cao, L.~Fan, G.~Phillips, and S.~Shenker, ``Web caching and
  zipf-like distributions: evidence and implications,'' in \emph{1999 IEEE
  Intl. Conf. on Computer Commun. (INFOCOM)}, vol.~1, Mar. 1999, pp. 126--134
  vol.1.

\bibitem{paulraj2003introduction}
A.~i~{Paulraj}, R.~{Nabar}, and D.~{Gore}, \emph{Introduction to Space-Time
  Wireless Communications}.\hskip 1em plus 0.5em minus 0.4em\relax Cambridge,
  U.K.: Cambridge Univ. Press, 2003.

\bibitem{Bertsekas2000Dynamic}
D.~P. Bertsekas, \emph{Dynamic Programming and Optimal Control}.\hskip 1em plus
  0.5em minus 0.4em\relax Belmont, MA, USA: Athena Scientific, 2000.

\bibitem{AnLiu2014}
A.~Liu and V.~K.~N. Lau, ``Cache-enabled opportunistic cooperative {MIMO} for
  video streaming in wireless systems,'' \emph{IEEE Trans. Signal Process.},
  vol.~62, no.~2, pp. 390--402, Jan. 2014.

\bibitem{W}
U.~Tamm, ``Some refelections about the {Lambert W} function as inverse of
  x*log(x),'' in \emph{2014 Information Theory and Applications Workshop
  (ITA)}, Feb. 2014, pp. 1--4.

\bibitem{WANG2017}
B.~Lv, L.~Huang, and R.~Wang, ``Cellular offloading via downlink cache
  placement,'' in \emph{2018 IEEE Intl. Conf. on Commun. (ICC)}, May 2018, pp.
  1--7.

\end{thebibliography}

\end{document}